\documentclass{article}
\usepackage[english]{babel}
\usepackage[letterpaper,top=2cm,bottom=2cm,left=3cm,right=3cm,marginparwidth=1.75cm]{geometry}
\usepackage{amsfonts}
\usepackage{amsmath}
\usepackage{amsthm}
\usepackage{graphicx}
\usepackage{tikz}
\usetikzlibrary{arrows.meta,positioning,calc}
\usepackage{needspace}
\usepackage{longtable,array}
\usepackage{microtype}
\usepackage{xcolor}
\usepackage{enumitem}
\usepackage[colorlinks=true, allcolors=blue]{hyperref}
\usepackage[numbers,sort&compress]{natbib}

\newtheorem{thm}{Theorem}[section]
\newtheorem{prop}[thm]{Proposition}
\newtheorem{lem}[thm]{Lemma}

\newcommand{\E}{\mathop{\mathbb{E}}}
\newcommand{\mc}{\mathcal}

\newcommand{\Prob}{\mathop{\mathbb{P}}}

\newcommand{\tail}{\operatorname{tail}}
\newcommand{\KL}{\operatorname{KL}}
\newcommand{\ind}{\mathbf{1}}
\newcommand{\keep}{\mathrm{keep}}
\newcommand{\del}{\mathrm{del}}

\title{No Truly Sublogarithmic-Round Algorithm for Distributed MIS in the LOCAL Model}
\author{
  Seri Khoury\thanks{Cognifence. \texttt{seri.khoury@cognifence.ai}}
  \and
  Aaron Schild\thanks{Google Research. \texttt{aaron.schild@gmail.com}}
}
\hypersetup{pdftitle={No Truly Sublogarithmic-Round Algorithm for Distributed MIS in the LOCAL Model},
pdfauthor={Seri Khoury, Aaron Schild}}
\begin{document}
\maketitle

\begin{abstract}
We prove that $o(\frac{\log n}{\log\log n})$-round randomized LOCAL algorithms find a maximal independent set (MIS) with subconstant probability on some input. This shows that the 40-year-old $O(\log n)$-round randomized CONGEST algorithm of Luby and Alon, Babai, and Itai is nearly optimal. Our lower bound improves upon the $\Omega(\sqrt{\log n})$-round lower bound of Khoury and Schild and nearly closes the 20+ year-old gap between the $O(\log n)$-round algorithm and the $\Omega(\sqrt{\frac{\log n}{\log\log n}})$-round lower bound of Kuhn, Moscibroda, and Wattenhofer.

To close this gap, one must show a lower bound for a graph that is not locally treelike. We do this by starting with a $D$-regular, $\Theta(\log_D n)$-girth graph, with $D = \Theta(\log n)$. Crucially, this is not our final graph. We obtain our hard instance by \emph{replacing} each vertex of this treelike graph with a clique on $n^c$ vertices for some constant $c \in (0,1)$. We then replace each edge with a bipartite Erd\H{o}s--R\'enyi graph between the corresponding cliques, with each edge present independently with probability $1/2$.

Intuitively, finding an MIS is challenging in this graph because it blocks the classic shattering phenomenon that underpins many algorithms in distributed computing. Roughly speaking, one round of Luby/ABI does not appreciably affect the structure of our hard instance. A clique is occupied by an independent set vertex with probability $\Theta(1/D)$, causing only an \emph{additive} decrease in the maximum degree of the connectivity graph between cliques. This is in contrast with its effect on $\Delta$-regular trees, where it achieves a constant-factor \emph{multiplicative} reduction in the maximum degree of the tree.

We extend this intuition to arbitrary LOCAL algorithms via a novel round elimination approach on our non-treelike hard instance. We describe algorithms that output a set of vertices in this graph using two metrics: their \emph{occupancy probability} and their \emph{conflict probability}. Our round elimination framework takes an $r$-round algorithm and produces an $(r-1)$-round algorithm where occupancy and conflict probabilities collectively slightly degrade, but individually can significantly improve or degrade. This round elimination result is obtained by exploiting the lack of large biclique complements in dense bipartite Erd\H{o}s--R\'enyi graphs, which in turn allows us to bound the entropy of the output distribution for any independent set algorithm.
\end{abstract}

\newpage

\tableofcontents

\newpage

\section{Introduction}
\label{sec:introduction}

In the LOCAL model of distributed computing, each vertex is an independent processor
that is given (a) an $O(\log n)$-bit node ID and (b) an infinite tape of random bits,
where tapes are independent for different vertices and IDs are distinct.
There are $r$ communication rounds, after which all vertices have received
all information from vertices within distance $r$. The vertices
must solve the problem of interest using only this local information.

One of the most extensively studied problems in distributed computing is the problem of
finding a maximal independent set (MIS): a set of pairwise nonadjacent
vertices that cannot be enlarged by adding another vertex ~\cite{Luby86,AlonBI86,RozhonG20,2013Barenboim,BalliuBHORS21,Ghaffari16,BarenboimEPS16,KuhnMW16,BarenboimE10,Linial92,Naor91,0001GHIR23,0001G23,GhaffariP19,GG24,BalliuBO22,Balliu0KO22,Balliu0KO21,LenzenW11,MetivierRSZ11,wattenhofer2020mastering,SchneiderW08}. Classical
randomized algorithms of Luby~\cite{Luby86} and Alon, Babai, and
Itai~\cite{AlonBI86} from 1986 find an MIS in $O(\log n)$ communication rounds with high
probability in the LOCAL model, where $n$ is the number of vertices.
Although later algorithms substantially improved the runtime on special graph classes~\cite{BarenboimEPS16,Ghaffari16,KS26},
$O(\log n)$ remains the best known worst-case bound in terms of $n$ alone
for general graphs. Kuhn, Moscibroda, and Wattenhofer~\cite{KuhnMW04,KuhnMW16}
proved an $\Omega(\sqrt{\log n/\log\log n})$ randomized lower bound.
Recently, Khoury and Schild~\cite{KS25,BCDO26} improved it to
$\Omega(\sqrt{\log n})$, using a lower bound for maximal matching on
trees and a reduction to MIS.

Despite 40 years of progress, whether MIS admits an $o(\log n)$-round algorithm
on general graphs remains open. In this work, we nearly close the gap: our lower bound
rules out $o(\log n/\log\log n)$ rounds and, in particular,
$O(\log^{1-\varepsilon}n)$ rounds for every fixed $\varepsilon>0$.
It leaves a factor of $\log\log n$ between the upper and lower bounds.

\begin{thm}[Randomized LOCAL lower bound]
\label{thm:local-main}
For every fixed $\sigma\in(0,1]$, there are constants $c_\sigma>0$
and $n_\sigma$ such that the following holds for every $n\ge n_\sigma$.
Any randomized LOCAL algorithm that outputs an MIS with probability
at least $\sigma$ on every $n$-vertex graph requires at least
\[
 c_\sigma\frac{\log n}{\log\log n}
\]
rounds in the worst case.
\end{thm}

This holds even when identifiers are the
numbers in $[n]$, the algorithm knows $n$ exactly, and the algorithm knows the maximum degree up to a constant factor.\footnote{Approximation is used here for a technicality: degrees are random since the graph is random.}

\paragraph{The high-girth barrier.}
The vast majority of lower bounds in the LOCAL model are proven for graphs
that are locally treelike for the following reason.\footnote{The prior $\Omega(\sqrt{\log n})$ lower bound for MIS applies for line graphs of trees, which are not trees, but was obtained by studying maximal matching on trees.}
Consider the endpoints of an arbitrary tree edge.
Fixing the common information for those endpoints leaves separate branches that
do not interact. This structure underlies both the
indistinguishability constructions of Kuhn, Moscibroda, and
Wattenhofer~\cite{KuhnMW16} and many round-elimination arguments. Furthermore,
Ghaffari~\cite{Ghaffari16,BarenboimEPS12,LenzenW11} gave a $O(\sqrt{\log n})$-round MIS algorithm
for graphs with girth $\Omega(\sqrt{\log n})$. Consequently, a stronger
MIS lower bound cannot be obtained by requiring every relevant LOCAL
neighborhood to be a tree. Thus, a new hard instance is required.

\paragraph{Overcoming the high-girth barrier.}
We construct our hard instance in two steps. First, we retain a
regular high-girth graph $H$ as a \emph{skeleton}, with degree $D = \Theta(\log n)$. Second, we replace every
skeleton vertex $v$ by a clique $C_v$ of size $k$, and replace each
skeleton edge $uv$ by a bipartite Erd\H{o}s--R\'enyi
graph with edge probability $1/2$.
The resulting graph contains triangles inside every clique and many
other short cycles. Neighborhoods of adjacent vertices can therefore
overlap heavily, even though the skeleton neighborhoods are trees.

Intuitively, this graph is hard because it blocks the shattering phenomenon
present in many LOCAL algorithms (e.g. \cite{Ghaffari16, BarenboimEPS12}).
Applying the first $t$ rounds only reduces the degree of the skeleton graph
by $O(t)$ additively, requiring $\Theta(D)$ rounds
to consume the entire graph. Arbitrary LOCAL algorithms could do clever operations
that consume more degree per round. Theorem \ref{thm:local-main} shows that
this is not possible.

\paragraph{The occupancy theorem.}
We prove Theorem \ref{thm:local-main} via a reduction to a statement
(Theorem \ref{thm:main}) that involves simpler types of algorithms that
are essentially equivalent to LOCAL algorithms on our hard instance.

For this part, the skeleton is the infinite rooted $D$-regular tree, and each
clique has named vertices $[k]$, which are also called \emph{coordinates}.
Let $G_{uv}\in\{0,1\}^{k\times k}$
be the bipartite adjacency matrix at an oriented skeleton edge, with
$G_{vu}=G_{uv}^{\mathsf T}$. Entries on distinct undirected edges
and within each matrix are independent fair bits.

A radius-$r$ rule sees the matrices on edges whose endpoints are both
within skeleton distance $r$ of its root. It returns a coordinate of
its root clique, or $0$ for no selection. The same rule is used everywhere
and is unchanged by rooted isomorphisms preserving matrix entries and
coordinate names. In particular, it has no skeleton identifiers or
supplied order of branches/port numbering. It may use a private final draw; these draws
are independent at different roots after all matrices are fixed.
Section~\ref{sec:model} gives the precise model.
For adjacent skeleton vertices $u,v$ and rule $\mc A$, write
\[
 p(\mc A)=\Prob[\mc A_v\ne0],\qquad
 \epsilon(\mc A)=
 \Prob[\mc A_u,\mc A_v\ne0,\ G_{uv}(\mc A_u,\mc A_v)=1].
\]
where $\mc A_w$ denotes the output of $\mc A$ on $w$.
These are the \emph{occupancy} and \emph{conflict probability} of $\mc A$ respectively; symmetry
makes them independent of the specified root or edge.
An independent set rule has $\epsilon(\mc A)=0$. This is our main technical result:

\begin{thm}[Occupancy bound]
\label{thm:main}
There is an absolute constant $D_0$ such that, for integers $D\ge D_0$,
$k$ with $D\le\log k\le D^2$, and $r\ge 1$, every symmetric radius-$r$
rule of Section~\ref{sec:model} satisfying $\epsilon(\mc A)=0$
obeys
\begin{equation}
 \boxed{\quad
 p(\mc A)\le\min\left\{1,\frac{5000r\log D}{D}\right\}.
 \quad}
 \label{eq:main-independent}
\end{equation}
\end{thm}

In particular, occupancy probability tends to zero whenever
$r=o(D/\log D)=o((\log n)/(\log \log n))$. By a union bound, the probability that a specified clique can be
dominated by at most one vertex from each neighboring clique is at most
$(k+1)^D\exp(-k2^{-D}) \le 1/\text{poly}(n)$.
Thus, with high probability, every MIS intersects every
clique, so an algorithm for MIS with success probability at least $\beta$
has occupancy probability at least 
$(1 - 1/\text{poly}(n))\beta$. This yields a contradiction whenever
$r = o((\log n)/(\log \log n))$.

Section~\ref{sec:local-reduction} proves Theorem~\ref{thm:local-main}
using Theorem~\ref{thm:main} as a black box. The remaining main-text
sections prove Theorem~\ref{thm:main}, beginning with a complete
one-round warmup. The appendix supplies the details of the LOCAL reduction.

We used AI in the production of this work.
We discuss our use of AI immediately before the references.

\subsection{Related Work}

Round elimination replaces an algorithm using $r$ rounds by one using
$r-1$ rounds for a suitably transformed problem. The lower bound for the
distributed Lov\'asz local lemma by Brandt, Fischer, Hirvonen, Keller,
Lempi\"ainen, Rybicki, Suomela, and Uitto~\cite{BrandtFHKLRSU16} developed
this method through sinkless orientation. Brandt~\cite{Brandt19}
subsequently gave an automatic transformation for locally checkable
problems on sufficiently high-girth graphs. Balliu, Brandt, Hirvonen,
Olivetti, Rabie, and Suomela~\cite{BalliuBHORS21} used it to prove tight
dependence on the maximum degree for MIS and maximal matching in the
relevant small-degree regimes. The also showed the result we have in this
work for deterministic LOCAL algorithms.

Further applications include exact degree dependence for bipartite
maximal matching~\cite{BO20}, lower bounds for ruling
sets~\cite{BalliuBO22}, MIS and relaxed dominating-set problems in
trees~\cite{Balliu0KO21}, and coloring~\cite{Balliu0KO22}. Khoury and
Schild~\cite{KS25} developed a self-reduction for matching that
preserves the problem and tracks the probability that a vertex remains
unmatched. Our occupancy argument also removes information from local
views, but its estimates concern coordinate distributions and conflicts
across random matrices. Balliu, d'Amore, and Olivetti~\cite{BDO26}
recently generalized self-reduction through the classical
round-elimination framework by changing the error quantity being
measured.

In this paper, we do extensive tracking of how much importance
certain skeleton edges have to the final output of the algorithm.
These importances are quantified using entropy, which was partly
inspired by the work of Backhausz, Gerencs\'er, and Harangi~\cite{BGH19}.

\subsection{Open Problems}

While this paper nearly closes an important gap in our understanding of the complexity for
MIS, many fundamental open problems related to MIS remain:

\paragraph{Problem 1:} \emph{What is the exact round complexity of MIS in the randomized LOCAL model as a function of $n$?}

In this paper, we give a lower bound of $\Omega((\log n)/(\log \log n))$ on this complexity.
The classical $O(\log n)$ round algorithm \cite{Luby86,AlonBI86} remains unbroken in spite of being 40 years old.
The precise complexity could be between these two extremes \cite{KS26}.

\paragraph{Problem 2:} \emph{What is the complexity of MIS in other distributed computing models?}

While there are many distributed computing models, two are particularly relevant to this question: quantum LOCAL and deterministic LOCAL. It is unclear if this lower bound applies to quantum LOCAL, even on our hard instance. For deterministic LOCAL, an $\Omega((\log n)/(\log \log n))$-round hardness result was already known prior to this work \cite{BalliuBHORS21} and an $\tilde{O}(\log^{5/3} n)$-round algorithm is known \cite{GG24}.

\paragraph{Problem 3:} \emph{What is the complexity of maximal matching and vertex cover in the randomized LOCAL model?}

The prior KMW lower bound \cite{KuhnMW04} applies to both $O(1)$-approximate vertex cover and maximal matching, while the lower bound of Khoury and Schild \cite{KS25} only applies to maximal matching. The lower bound in this paper applies to neither problem. Thus, the $\tilde{\Theta}(\sqrt{\log n})$ gap remains for both problems. 

\section{Technical Overview}
\label{sec:overview}

We now outline the techniques used to prove Theorem \ref{thm:local-main}.
We break up the discussion into three parts: describing the hard instance,
the occupancy argument (Theorem \ref{thm:main}), and the proof of
the desired result (Theorem \ref{thm:local-main}) given Theorem \ref{thm:main}.
At the end, we give a detailed outline of the paper.

Note that the occupancy discussion precedes the reduction discussion in
the technical overview, but the opposite is the case for the corresponding sections.
We chose to discuss the occupancy theorem first as it contains
all of the novelty in our argument. The reduction uses standard techniques.

\subsection{The Hard Instance}

Our hard instance is constructed as follows:

\begin{enumerate}[itemsep=1pt]
    \item Use Lemma \ref{lem:local-girth} to construct a $D$-regular graph $H$ with girth at least $(\log_D n)/100$. This constructs $H$ by arguing that a edge-count-maximizing $n$-vertex graph with the required girth and maximum degree $D$ is actually a regular graph with degree $D$. Alternatively, one can sample a random regular graph and argue that it has the required girth with positive probability.
    \item Replace each vertex with a clique of size $k = n^c$ for some constant $c\in (0,1)$.
    \item Replace each edge with a random bipartite graph, with each edge sampled independently with probability 1/2.
\end{enumerate}

Running one round of the classical algorithm \cite{Luby86,AlonBI86} on this graph has the following effects, roughly speaking:

\begin{enumerate}[itemsep=1pt]
    \item A $1/D$-fraction of cliques contain a selected vertex in expectation.
    \item Each of those cliques gets deleted.
    \item Every other clique reduces in size by a factor of roughly $2^{\ell}$, where $\ell$ is the number of neighboring cliques that contain a selected vertex.
\end{enumerate}

Roughly speaking, $\ell = \Theta(1)$ and a large fraction of cliques have the same value of $\ell$. Thus, applying one round of Luby/ABI roughly has the effect of (a) deleting a $1/D$ fraction of the cliques, reducing a typical skeleton vertex's degree by $D \times 1/D = 1$ and (b) reducing the size of most other cliques by the same constant factor. In particular, the graph is now a copy of the original, with $k$ decreased multiplicatively and $D$ decreased additively. Since $k \gg 2^{\Theta(D)}$, the algorithm must take $\Omega(D)$ iterations.

General LOCAL algorithms may have too many dependencies for this exact argument to work. We instead describe an argument based on round elimination that works by tracking occupancy probability and conflict probability.

\subsection{The Occupancy Argument}

Recall that an $r$-round rule consumes the matrices (random graphs) for each edge within distance $r$ of a skeleton vertex $v$ and outputs an integer between 0 and $k$ inclusive, with 0 output if the rule picks no vertex from $C_v$. We build up to the final argument by considering the following steps, all of which are fully proven in our writeup:\footnote{Some of the propositions and lemmas in this work are only used in the warmup. In all cases, they are simpler versions of results that are used in the final proof.}

\begin{enumerate}[itemsep=1pt]
    \item A standalone 1-round impossibility result for our hard instance
    \item A roughly $\Omega(D^{1/3})$-round impossibility result using a weak round elimination result and a weak final bound
    \item A roughly $\Omega(D^{1/2})$-round impossibility result using the same round elimination result, but an optimal final bound
    \item The final $\Omega(D/(\log D))$-round impossibility result using a better round elimination result
\end{enumerate}

Each of these arguments contains variations of the same three steps, which we initially point out during the one-round warmup discussion.

\subsubsection{The One Round Warmup}

Consider a one round rule; i.e. a rule that consumes the $D$ matrices for
skeleton edges incident with a root and produces a number between 0 and $k$ inclusive.
The rule has zero conflict probability, so never outputs two adjacent vertices in
two adjacent cliques. We analyze its behavior in three steps.

\paragraph{1. Short list given one fixed incident matrix (Lemma \ref{lem:matrix}):}
At one round, adjacent cliques see a common matrix $G_{uv}$.
All their other input matrices are independent.
Therefore, after fixing $G_{uv}$, their outputs are independent
random choices. If each endpoint still has many possible nonzero
outputs, the possible choices contain a large rectangle of matrix
entries. A typical random matrix has an edge in every such
rectangle, essentially because $k$-vertex dense Erd\H{o}s--R\'enyi graphs
have maximum clique size $O(\log k)$. Since every pair of possible choices has positive
conditional probability due to independence, this would create a conflict.
Consequently, for most matrices, at least one endpoint has only
$O(\log k)$ possible outputs.

\paragraph{2. Short list of possible outputs overall (Proposition \ref{prop:warm-list}):}
The endpoint with few possible outputs may depend on the matrix.
Symmetry nevertheless says that a random incident matrix gives
the root a short list with probability at least $1/2-o(1)$.
The $D$ incident matrices are independent, so the resulting
lists are independent. Every list that applies must contain the
root's realized nonzero output.
A coordinate that appears in many independent short lists must,
with high probability, belong to a single fixed short list (Lemma \ref{lem:frequent-labels}).
This gives a fixed list of $D^{O(1)}$ coordinates containing
almost every output.

\paragraph{3. Short list implies low occupancy via entropy (Lemma \ref{lem:warm-forget}):}
A rule could concentrate on one coordinate, so list size alone
does not bound occupancy. The additional observation is that
the output has only a limited amount of information about its
$D$ independent incident matrices (Lemma \ref{lem:product-information}).
Entropy measures this information: for probabilities
$p_1,\ldots,p_s$, it is $\sum_i p_i\log(1/p_i)$.
Independent inputs cannot each reveal a large fraction of the
same limited information budget.
By symmetry, a specified incident matrix accounts for at most
a $1/D$ fraction of that budget.

Imagine hiding every matrix and drawing a coordinate from the
rule's unconditional output probabilities. At two endpoints
these new draws are independent of the shared matrix between them, so their
conflict probability is exactly half the square of their
occupancy. Hiding the matrix can increase conflict only by an
amount controlled by the information it carried about the
outputs (Lemma \ref{lem:sqrt-event}). A short list makes
that amount small (Lemma \ref{lem:occupancy-entropy}). Combining the
two estimates gives occupancy $O(\log D/D)$.
Section~\ref{sec:warmup} proves this entire one-round argument.

\subsubsection{$\tilde{\Omega}(D^{1/3})$-Round Lower Bound via Round Elimination}

We remove the outermost input layer one round at a time.
For each fixed smaller view, keep its most probable coordinates
and change other outputs to $0$. We call this operation
\emph{truncation}. Then average the kept output probabilities
over the hidden layer and draw independently from those
averaged probabilities. Effectively, we are moving the randomness
from the outermost shell into the center. This idea is called
\emph{conditional resampling}. This defines a rule using one fewer
round and preserves the kept occupancy.

There are two quantities to control. Truncation loses occupancy,
and conditional resampling can introduce conflicts. Thus after the first
averaging step we must handle rules with small positive conflict
probability, even though the original rule has zero conflict.
Exact lists are no longer adequate: an event of tiny probability
can make the set of possible outputs enormous.
We use a short list together with the probability of falling
outside it. Section~\ref{sec:truncation} proves that a deleted
part of mass $\tau$ and conflict probability $w$ satisfies
\[
 \tau\le b+6\sqrt w,\qquad b=\frac{2100\log D}{D}.
\]
Here mass means simply the probability of a nonzero output. $\sqrt{w}$
is able to bound $\tau$ thanks to the fact that for 0-round
algorithms, conflict probability is quadratic in occupancy probability.
Thus, taking the square roots makes the ``units'' agree. We prove
this bound (Proposition \ref{prop:truncation}) in three steps:

\paragraph{1. Upper-bounding edge-neighborhood tails by conflict
probability (Proposition \ref{prop:endpoint-tails}):} This extends
the one round argument from zero conflict to positive conflict
rules. The main difference is that the notion of a shortlist is different:
we upper bound the amount of mass outside of a shortlist (called a \emph{tail})
conditioned on a common edge-neighborhood (union of neighborhoods of endpoints). This
once again uses the lack of large cliques in Erd\H{o}s--R\'enyi graphs.

\paragraph{2. Upper-bounding vertex-neighborhood tails by edge-neighborhood tails (Lemma \ref{lem:branch-transfer}):}
Each edge-neighborhood tail corresponds to a short list of outputs, with different
short lists for each of the branches. As in the one round argument, we have a bunch of
independently-chosen short lists that must intersect. This once again guarantees a short
merged list, but this time also provides a bound on the tail of that combined short list.

\paragraph{3a. Lower-bounding vertex-neighborhood tails by occupancy probability (Lemma \ref{lem:scalar-entropy}):}
The expectation of a product of two vertex-neighborhood-derived tails is
close to the product of expectations thanks to the near-independence
of adjacent vertices. An entropy-based argument occurring in
Section \ref{sec:tree-entropy} that builds on Lemma \ref{lem:product-information} shows this.

\paragraph{3b. Upper-bounding conflict probability using occupancy probability:}
This uses Proposition \ref{prop:one-step}, which in turn uses Lemma \ref{lem:product-information} to obtain the required similarity between the conditionally-resampled output distribution and the original one.

\paragraph{Bounding losses:}
As rounds are eliminated, occupancy probability tends to decrease,
while conflict probability tends to increase.
The loss estimate involves $\sqrt w$. To sum losses, we need a
bound on the \emph{sum} of the deleted conflict probabilities.
Section~\ref{sec:sequential} first gives a direct, deliberately
coarse calculation: charge the full possible information cost
at each reduction. This gives a cubic term in $r$.

\subsubsection{$\tilde{\Omega}(D^{1/2})$-Round Lower Bound}
Next, we track the output entropy already created by earlier
averaging steps. It cannot be charged anew at every step; it
can disappear only when outputs are deleted. This gives a
quadratic term via a telescoping sum. No other part of the argument changes.

\subsubsection{The Main $\Omega(D/(\log D))$-Round Lower Bound}

For the final argument, we only modify the third part of the prior arguments.
Recall that the round elimination procedure takes the original
$r$-round rule and repeatedly alternates between truncating and conditionally
resampling to turn it into lower-radius rules. This can be replaced with
a strategy that does all truncations first, followed by
one combined conditional resampling step (Lemma \ref{lem:ancestors}).
Qualitatively, this is more efficient as each truncated output is only truncated
at one step and the loss can now be attributed only to that step
as opposed to the multiple prior steps.
Quantitatively, we get to use a more efficient multi-round analogue of
Proposition \ref{prop:one-step} (Proposition \ref{prop:direct-projection}) to replace Step (3b).
This completes the proof of the occupancy theorem.

\subsection{From the Occupancy Theorem to LOCAL}

In the LOCAL model~\cite{Linial92}, vertices communicate synchronously along graph
edges, with no restriction on message size or local computation.
Every vertex runs the same program, has a distinct identifier and local
port numbers, and may know $n$ exactly. Private random tapes are
independent between vertices; they may be infinite. We count a fixed
worst-case number of communication rounds. Success means that the
\emph{entire} output is an MIS, rather than that each vertex is correct
with a fixed probability. We allow degrees and neighbor identifiers to
be known initially, which only strengthens the algorithms against
which the lower bound holds.

Theorem~\ref{thm:main} concerns a rule whose only inputs are
random inter-clique matrices. A LOCAL algorithm also sees
identifiers and receives private randomness. We will turn it
into a collection of edge-input rules without changing the matrix
distribution. A Yao minimax-style argument gives the desired result.

\subsubsection*{Preparing Vertex IDs}

Most of our task amounts to preparing the input to the LOCAL algorithm.
First, we prepare vertex IDs. Each clique contains $k$
vertices, and we must assign IDs to those vertices that are unique
and different from IDs for vertices in any other clique. Note that
a rule's only access to randomness (i.e. symmetry-breaking) is
through the skeleton edges.

Thus, we extract vertex IDs from skeleton edges as follows. For simplicity,
assume that $k$ (the clique size) is even. Split
the set of possible coordinate names $[k]$ in each clique into two fixed
halves $\{1,2,\hdots,k/2\}$ and $\{k/2+1,\hdots,k\}$.
For a skeleton vertex $v$ and a neighbor $u$, produce a bit string
of length $k^2/4$ by concatenating the (bit) entries $G_{vu}(i,j)$
with $i\le k/2$ and $j \ge k/2+1$ (i.e. first half rows to second half columns) in lexicographic order.
Then, obtain a single length $k^2/4$-bit string by XORing the strings for all neighbors of $v$.
Note that this string does not require a port numbering to produce.
Finally, break up this string into chunks with a fixed $b_{ID} = \Theta(\log n)$ length and
let the $i$th vertex in $C_v$ receive the $i$th chunk as its identifier. Some chunks may go unused.

For any pair of vertices $u,v$ and any $a\in C_u, b\in C_v$, the IDs of $a$ and $b$ are independent
random variables if $a\ne b$ thanks to the fact that the entries of $G_{vu}$ from the first half rows to second half columns
do not intersect with the entries of $G_{vu} = G_{uv}^T$ from the second half rows to the first half columns.
Thus, IDs are unique with high probability.\footnote{This argument assumes that IDs are only required to be at most $n^{O(1)}$. A more nuanced argument is required when the IDs are required to be in $[n]$. This argument is described in the appendix.} Note though that IDs are correlated with the input graph, which is allowed. 

\subsubsection*{Preparing the Private Random Strings}

Note that we said that we would produce a \emph{collection} of edge-input rules. We exploit this now.
Fix a list of infinite bitstrings, with one bitstring for each of the $2^{b_{ID}}$
possible vertex IDs. We will produce a separate edge-input rule for each possible list.\footnote{Note that there are uncountably many such lists, so the collection of obtained edge input rules is uncountable.}
We start by giving the vertex with ID $i$
access to the $i$th bitstring in the fixed list.

\subsubsection*{Obtaining a Collection of Zero-Conflict Edge-Input rules}

We now have unique vertex IDs for the vertices
and tapes (bitstrings) for each vertex, so the supplied LOCAL algorithm may be applied. This
results in one edge-input rule per list of bitstrings. The LOCAL algorithm
produces a set of vertices in the input graph, though this set
may not be an independent set thanks to the algorithm's failure probability.
We force the output set to be a (possibly not maximal) independent set
by deleting any vertex from the output with a neighbor also in the output.
This results in a collection of zero-conflict $r+3$-round edge-input rules,
where $r$ is the round complexity of the LOCAL algorithm.

\subsubsection*{Obtaining a Probability Contradiction}

Pick an arbitrary clique in the hard instance and apply each edge-input rule
to that clique and let $r = o((\log n)/(\log \log n))$. Then
Theorem \ref{thm:main} says that the occupancy probability is subconstant
for each rule. Taking the expectation over the collection implies
that the expected occupancy probability is subconstant. By the probabilistic method,
there exists a realization of the edges in the graph (and resultant vertex IDs) for which randomizing
over tapes causes occupation of the chosen clique with subconstant probability.
But this contradicts the fact that an MIS intersects each clique
with high probability, the deletion postprocessing does not modify an MIS,
and the LOCAL algorithm produces an MIS with constant probability for any input
graph with unique IDs. This is the desired contradiction.

\subsection{Outline and Reading Advice}

Section \ref{sec:local-reduction} reduces Theorem \ref{thm:local-main} to Theorem \ref{thm:main},
with some details deferred to the appendix. All other sections prove
Theorem \ref{thm:main}. Section \ref{sec:model} precisely defines $r$-round rules
and proves the fundamental facts about entropy that are used throughout the paper.
Section \ref{sec:warmup} proves a self-contained 1-round occupancy probability upper bound.
Section \ref{sec:one-step} starts building towards round elimination by introducing the concept of conditional resampling.
Section \ref{sec:tree-entropy} proves near-independence of random variables on neighboring vertices in very general settings that are useful for reasoning about multiple round eliminations at a time. Section \ref{sec:truncation} proves the main round elimination bound (Proposition \ref{prop:truncation}), effectively generalizing the techniques of Section \ref{sec:warmup}. Section \ref{sec:sequential} proves the two weaker lower bounds, while Sections \ref{sec:joint} and \ref{sec:main-proof} complete the final result.

The sections are ordered in terms of precedence: all later results only depend on earlier results, with the exception of the appendix. We recommend that the reader read the paper out-of-order. The reader should start with Section \ref{sec:local-reduction}. Then, they should skim the results of Section \ref{sec:model}, mostly to understand the definition of $r$-round rules.

A reader who wants to focus on the high-level structure of the argument should skip ahead to Section \ref{sec:one-step}, followed by the statement of Proposition \ref{prop:truncation}. Then, read Section \ref{sec:sequential} and skim Section \ref{sec:joint}. At this point, the reader may understand the high-level structure of the argument and should return to Section \ref{sec:warmup}. A reader who is more interested in techniques may prefer to start with Section \ref{sec:warmup} and then continue with the sections discussed for the high-level reader.

\section{From the Occupancy Bound to Randomized LOCAL}
\label{sec:local-reduction}

Fix a $T$-round LOCAL algorithm $\mc A$ for $n$ vertices.
The algorithm may know $n$, use independent private random tapes,
and see neighbor identifiers initially. We use legal port numbers
obtained by sorting each vertex's neighbors by identifier.
Correctness on every legal input includes these port assignments.
For $w\in V(H)$, $B_t^H(w)$ denotes the closed skeleton ball
of radius $t$ about $w$.

\begin{prop}[Finite LOCAL obstruction]
\label{prop:mis-obstruction}
Let $H$ be a simple $D$-regular graph on $N$ vertices, with
$kN\le n$ and girth at least $2T+8$.
Suppose $D\ge D_0$, $D\le\log k\le D^2$, $k$ is even, and
$k b_{\mathrm{ID}}\le\lfloor k^2/4\rfloor$, where $b_{\mathrm{ID}}=\lceil\log_2 n\rceil$.
Set
\[
 L=k\left(1+D\sum_{i=0}^{T+1}(D-1)^i\right),\qquad
 \delta_{\mathrm{ID}}=\frac{L(L-1)}n,\qquad
 \gamma_{\mathrm{EXT}}=(k+1)^D\exp(-k2^{-D}).
\]
If $\mc A$ finds an MIS with probability at least $\sigma>0$
on every legal $n$-vertex input, then
\begin{equation}
 \sigma(1-\delta_{\mathrm{ID}}-\gamma_{\mathrm{EXT}})
 \le \min\left\{1,\frac{5000(T+3)\log D}{D}\right\}.
 \label{eq:local-obstruction}
\end{equation}
\end{prop}
\begin{proof}
Replace every skeleton vertex by a $k$-clique, use independent
fair edge matrices between neighboring cliques, and pad with
isolated vertices to obtain exactly $n$ vertices.
Use Lemma~\ref{lem:local-identifiers} to produce vertex IDs.
This lemma works as follows. Partition clique coordinates
into two parts $P = \{1,2,\hdots,k/2\}$ and $Q = \{k/2+1,\hdots,k-1,k\}$.
For the oriented edge $v\to u$, assign the matrix block
$G_{vu}[P,Q]$ to $v$; the reverse direction uses the disjoint
block $G_{uv}[P,Q] = (G_{vu}[Q,P])^T$. XOR the assigned strings over all edges
incident to $v$, and split the result into $k$ strings of $b_{\mathrm{ID}}$ bits.
Reducing each string modulo $n$ and adding $1$ gives tentative identifiers in
$[n]$. They are mutually independent, with point probabilities
at most $2/n$, because different vertices use disjoint random
bits. They can still be correlated with the graph.

Independently sample one private tape $W_i$ for every identifier
$i\in[n]$, and fix the entire table. Apply Lemma~\ref{lem:local-simulation}
to this table to obtain a radius $(T+3)$-rule. This works as follows. A vertex with tentative
identifier $i$ uses $W_i$. For every clique, simulate $\mc A$
only if all identifiers in its radius-$(T+1)$ skeleton ball
are distinct; otherwise produce no candidates. Retain a
candidate only when none of its graph neighbors is a candidate.
This produces a deterministic, symmetric, always-independent
edge-input rule of radius $T+3$. Its occupancy at a fixed
clique $C_v$ is bounded by Theorem~\ref{thm:main}.
The bound holds for every fixed tape table, so it also holds
after averaging over that table.

With probability at least $1-\delta_{\mathrm{ID}}$, the tentative identifiers
in $B_{T+2}^H(v)$ are all distinct. On this event, preserve these
identifiers and complete them to a bijection on all $n$ vertices.
The completion depends on the matrices, not on the tapes.
The candidates used to repair $C_v$ agree with the outputs of
$\mc A$ on this legal input. Distinct identifiers access distinct
independent tapes, so the assumed success probability is at least
$\sigma$ on each such fixed input
(Lemma~\ref{lem:local-completion}).

Except on a graph event of probability at most $\gamma_{\mathrm{EXT}}$, no
choice of at most one vertex from each neighboring clique can
dominate $C_v$. Thus every MIS occupies $C_v$
(Lemma~\ref{lem:local-domination}). On the intersection of
these two good graph events, a successful execution occupies
$C_v$ and survives conflict removal. Averaging first over the
tapes and then over the matrices gives occupancy at least
$\sigma(1-\delta_{\mathrm{ID}}-\gamma_{\mathrm{EXT}})$, proving~\eqref{eq:local-obstruction}.
\end{proof}

\begin{proof}[Proof of Theorem~\ref{thm:local-main}]
Fix the success guarantee $0<\sigma\le1$. For all sufficiently
large $n$, suppose
\[
 T<\frac{\sigma\log n}{100000\log\log n}.
\]
Choose
\[
 D=\left\lfloor\frac{\log n}{4}\right\rfloor,\quad
 k=2\lceil e^D/2\rceil,\quad
 N=2\left\lfloor\frac{n}{2k}\right\rfloor,\quad g=2T+8.
\]
Lemma~\ref{lem:local-girth} supplies a $D$-regular graph on
these $N$ vertices with girth at least $g$:
its sufficient size is $8D^{g-2}$, whereas
$N=n^{3/4+o(1)}$ and
$D^{g-2}\le n^{2\sigma/100000+o(1)}$.
The remaining hypotheses of Proposition~\ref{prop:mis-obstruction}
hold for large $n$ (Lemma~\ref{lem:local-parameters}). Moreover,
\[
 \delta_{\mathrm{ID}}\le\frac{(2kD^{T+2})^2}{n}
 \le n^{-1/2+2\sigma/100000+o(1)}=o(1),
\]
and
\[
 \log\gamma_{\mathrm{EXT}}\le D(D+1)-e^{(1-\log2)D}\longrightarrow-\infty.
\]
Its left-hand side is therefore $\sigma-o(1)$.
Since $D\ge(\log n)/5$ and $\log D\le\log\log n$ for
large $n$, its right-hand side is at most $\sigma/4+o(1)$,
a contradiction. All thresholds depend only on $\sigma$ and
the absolute occupancy constants, not on the algorithm.
This proves the claimed worst-case lower bound, including
algorithms that know $n$ and use the identifier space $[n]$.
For $\sigma=2/3$, the same estimates with
$T<\log n/(100000\log\log n)$ give $2/3-o(1)$ on the left
and $1/4+o(1)$ on the right, establishing the stated special constant.
\end{proof}
\Needspace{5\baselineskip}
\section{Preliminaries}
\label{sec:model}

The rest of the paper proves Theorem \ref{thm:main}. In this section, we
discuss relevant definitions.

\subsection{Radius-$s$ Rules}

For a skeleton vertex $v$, let $E_s(v)$ contain the edges with
both endpoints in the closed radius-$s$ ball around $v$.
Write $N_s(v)$ for the tuple of matrices on these edges.
The underlying rooted tree is fixed; all randomness in a view
comes from the matrices. In particular, $N_0(v)$ is empty.
For conditioning, we name the positions of these independent
variables on a fixed tree. These names are bookkeeping for the
proof and are not inputs to a rule.

\paragraph{Symmetric local selection rule.}
\label{def:kernel}
A radius-$s$ rule $\mc A$ assigns to every view $y=N_s(v)$
numbers $\pi_0(y),\pi_1(y),\ldots,\pi_k(y)\ge0$ with sum $1$.
At $v$ it outputs $a\in\{0,\ldots,k\}$ with probability $\pi_a(y)$.
Conditional on all matrices, these output draws are independent
at different vertices.

The assignment is the same at every skeleton vertex and is
unchanged by rooted isomorphisms of the matrix-labeled tree
that preserve clique-coordinate names. This is the symmetry
assumption. A deterministic rule assigns probability $1$ to
one output. An independent-set rule has
$\epsilon(\mc A)=0$.

We use ordinary finite conditional probabilities throughout.
For example, $\Prob[Z=a\mid Y=y]$ is the probability of output
$a$ after fixing the input information $Y$ to be $y$.
The \emph{conditional distribution} is the list of these
probabilities for all $a$; it is sometimes called a conditional
law. We will write the probabilities explicitly when constructing
a new rule.
The \emph{support} of a distribution is the set of outputs
having positive probability.

\paragraph{Shared information and probability outside a list.}
\label{def:views}
For adjacent $u,v$ and $t\ge0$, put
\[
 F_t(uv)=\{G_{uv}\}\cup N_t(u)\cup N_t(v),
\]
where the union retains each matrix variable once.
For a random output $Z\in\{0\}\cup[k]$, define
\[
 \tail_a(Z\mid Y=y)
 =\min_{\substack{S\subseteq[k]\\|S|\le a}}
       \Prob[Z\notin S\cup\{0\}\mid Y=y].
\]
Without $Y$ this denotes the unconditional probability.
The minimizing list consists of the $a$ most probable nonzero
coordinates, or all coordinates when $a\ge k$.
Ties are always broken by coordinate order.
We also write $\tail_a(P)$ when $P$ is a specified probability
distribution.

The number $\tail_a$ measures the probability missed by the
best $a$-coordinate list. Output $0$ is included for free;
it is never a selected graph vertex.

We repeatedly use the following lemma, which uses the skeleton graph's high girth:

\begin{lem}[Product decomposition of local views]
\label{lem:views}
Let $s\ge1$.
At an edge $uv$, the shared view $F=F_{s-1}(uv)$ contains the
central matrix and both $N_{s-1}(u)$ and $N_{s-1}(v)$.
Conditional on $F$, the remaining data in $N_s(u)$ and $N_s(v)$
are independent. Consequently, the outputs of any radius-$s$
rule at $u,v$ are independent conditional on $F$.

At a root $v$, enumerate its neighbors as $v_1,\ldots,v_D$ only
for analysis, and put $Y=N_{s-1}(v)$. Let $X_i$ be the matrix data
in $F_{s-1}(vv_i)$ outside $Y$. Conditional on $Y$, the $X_i$ are
independent, their union supplies exactly the data in $N_s(v)$
outside $Y$, and $F_{s-1}(vv_i)$ is specified by $(Y,X_i)$.
\end{lem}

\begin{proof}
On a tree, the shared edge set is
\[
 E_s(u)\cap E_s(v)
 =\{uv\}\cup E_{s-1}(u)\cup E_{s-1}(v).
\]
The two edge sets outside this intersection are disjoint. All
edge matrices are independent, so conditioning on the shared
matrices leaves these two collections independent. The independent
private output draws preserve this conditional independence.

For the second assertion, the edges of $N_s(v)$ outside $Y$ are
the edges between depths $s-1$ and $s$. They partition according
to the first neighbor on the path from $v$. The part in branch
$i$ is exactly $X_i$, including the central edge $vv_i$ when
$s=1$. These parts are disjoint and independent of the inner
edge data. This proves every assertion.
\end{proof}

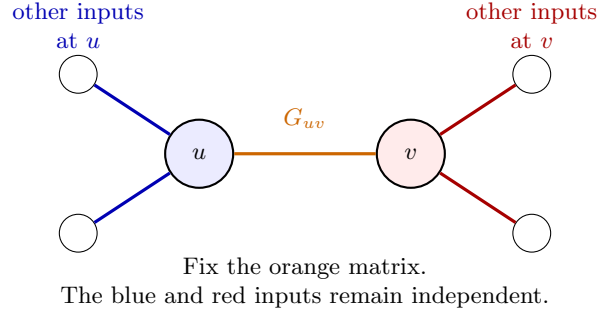
\begin{figure}[htbp]
\centering
\begin{tikzpicture}[x=1cm,y=1cm, every node/.style={font=\small}]
 \node[circle,draw,thick,fill=blue!8,minimum size=9mm] (u) at (-1.4,0) {$u$};
 \node[circle,draw,thick,fill=red!8,minimum size=9mm] (v) at (1.4,0) {$v$};
 \foreach \n/\x/\y in {a/-3/1.05,b/-3/-1.05,c/3/1.05,d/3/-1.05}
   \node[circle,draw,fill=white,minimum size=5mm] (\n) at (\x,\y) {};
 \draw[very thick,blue!70!black] (a)--(u)--(b);
 \draw[very thick,red!65!black] (c)--(v)--(d);
 \draw[very thick,orange!80!black] (u)--node[above=2mm] {$G_{uv}$}(v);
 \node[blue!70!black,align=center] at (-3,1.7) {other inputs\\at $u$};
 \node[red!65!black,align=center] at (3,1.7) {other inputs\\at $v$};
 \node[align=center] at (0,-1.7) {Fix the orange matrix.\\The blue and red inputs remain independent.};
\end{tikzpicture}
\caption{The shared information for two one-round views, shown
for $D=3$. Every circle denotes an entire clique, and every
line denotes its random bipartite matrix. At general radius,
fixing $F_{s-1}(uv)$ leaves the two remaining input regions
independent in the same way.}
\label{fig:shared}
\end{figure}

\subsection{Entropy}

For a finite random variable $Z$ with probabilities $p_z$, define
\[
 H(Z)=\sum_{z:p_z>0}p_z\log(1/p_z),\qquad
 H(Z\mid Y)=\sum_y\Prob[Y=y]H(Z\mid Y=y).
\]
All logarithms are natural in this paper unless otherwise mentioned.
An output with $m$ equally likely possibilities has entropy
$\log m$; an output determined by the available information
has conditional entropy zero.
Write $h(x)=-x\log x-(1-x)\log(1-x)$ for binary entropy, using
$0\log0=0$. Define
\[
 I(Z;X\mid Y)=H(Z\mid Y)-H(Z\mid X,Y).
\]
This \emph{conditional mutual information} is the amount by
which seeing $X$ reduces the uncertainty in $Z$, once $Y$ is known.

For probability distributions $P,Q$ on the same finite set,
their \emph{relative entropy} is
\[
 \KL(P\Vert Q)=\sum_{x:P(x)>0}P(x)\log\frac{P(x)}{Q(x)}.
\]
It is $+\infty$ if $Q(x)=0<P(x)$ for some $x$.
The inequality $\log z\le z-1$ shows $\KL(P\Vert Q)\ge0$.

Expanding definitions gives the entropy chain rule
$H(X,Z\mid Y)=H(X\mid Y)+H(Z\mid X,Y)$.
It also expresses $I(Z;X\mid Y)$ as the average relative
entropy of the conditional joint distribution from the product
of its conditional marginals. Thus $I\ge0$, or equivalently,
conditioning cannot increase entropy.

Comparing a distribution on at most $m$ values with the
uniform distribution gives $H\le\log m$.
These facts also give subadditivity and the chain rule for $I$.
\begin{lem}[Square-root event probabilities]
\label{lem:sqrt-event}
For probability distributions $P,Q$ on a finite set
and any event $E$,
\begin{equation}
 |\sqrt{P(E)}-\sqrt{Q(E)}|
 \le
 \left(\sum_x(\sqrt{P(x)}-\sqrt{Q(x)})^2\right)^{1/2}
 \le\sqrt{\KL(P\Vert Q)}.
 \label{eq:sqrt-event}
\end{equation}
\end{lem}

\begin{proof}
The first inequality is the reverse triangle
inequality for the square-root vectors restricted
to $E$, followed by enlarging the sum to all $x$.
If the relative entropy is infinite there is
nothing more to prove. Otherwise, put
$a=\sum_x\sqrt{P(x)Q(x)}$.
Jensen's inequality for the logarithm gives
\[
 \log a
 =\log\E_{x\sim P}\sqrt{\frac{Q(x)}{P(x)}}
 \ge-\frac12\KL(P\Vert Q).
\]
Thus $\KL(P\Vert Q)\ge-2\log a\ge2(1-a)$.
The squared middle expression in
\eqref{eq:sqrt-event} equals $2(1-a)$.
\end{proof}

We use square roots because the baseline conflict probability
will be quadratic in occupancy. This inequality converts a
relative-entropy bound directly into a bound on that scale.

We bound KL-divergence using the following conditional resampling-oriented
result that shifts randomness from the outer rim determining $A$ to the middle,
leaving a lower-radius neighborhood $Y_A$:

\begin{lem}[Relative entropy of conditional resampling]
\label{lem:resampling-kl}
Let $F,A,B$ be finitely-supported random variables, and let $Y_A,Y_B$
be functions of $F$. Let $P$ be the distribution of $(F,A,B)$.
Define $Q$ by keeping the same distribution of $F$ and,
conditional on $F=f$, drawing the two outputs independently
with probabilities
\[
 \Prob[A=a\mid Y_A=Y_A(f)],
 \qquad
 \Prob[B=b\mid Y_B=Y_B(f)].
\]
Then
\[
 \begin{aligned}
 \KL(P\Vert Q)
 &=H(A\mid Y_A)+H(B\mid Y_B)-H(A,B\mid F)\\
 &=I(A;F\mid Y_A)+I(B;F\mid Y_B)+I(A;B\mid F).
 \end{aligned}
\]
In particular, if $A,B$ are independent conditional on $F$,
then
\[
 \KL(P\Vert Q)
 =I(A;F\mid Y_A)+I(B;F\mid Y_B).
\]
This conditional independence holds whenever
$A=\alpha(F,U)$ and $B=\beta(F,V)$ for random variables
$U,V$ that are independent conditional on $F$.
\end{lem}

\begin{proof}
Write $p_f(a,b)=\Prob[A=a,B=b\mid F=f]$,
$q_A(a\mid y)=\Prob[A=a\mid Y_A=y]$, and similarly for $q_B$.
The finite support makes it so that these probabilities
are discrete rather than density functions.
If $P(f,a,b)>0$, then both $q_A(a\mid Y_A(f))$ and
$q_B(b\mid Y_B(f))$ are positive. Thus the likelihood ratio
is defined on the support of $P$. The common probability of
$F=f$ cancels, giving
\[
 \KL(P\Vert Q)
 =\E_P\!\left[
   \log p_F(A,B)
   -\log q_A(A\mid Y_A)
   -\log q_B(B\mid Y_B)
 \right].
\]
The three terms are
$-H(A,B\mid F)$, $H(A\mid Y_A)$, and $H(B\mid Y_B)$.
This proves the first identity.
Since $Y_A,Y_B$ are functions of $F$,
\[
 I(A;F\mid Y_A)=H(A\mid Y_A)-H(A\mid F),
 \qquad
 I(B;F\mid Y_B)=H(B\mid Y_B)-H(B\mid F).
\]
Together with
$I(A;B\mid F)=H(A\mid F)+H(B\mid F)-H(A,B\mid F)$,
these equalities prove the second identity.
Conditional independence gives $I(A;B\mid F)=0$.
Finally, functions of conditionally independent inputs
remain conditionally independent.
\end{proof}

Next, we bound the entropy of a distribution whose probability mass is both very
lopsided and supported on a small number of outputs:

\begin{lem}[Entropy of a rarely nonzero output]
\label{lem:occupancy-entropy}
Suppose the conditional distribution of $C\in\{0\}\cup[k]$ given $Y$ is supported on
at most $m$ nonzero labels and $0$, for an integer $m\ge1$.
If $p=\Prob[C\ne0]$, then
\begin{equation}
 H(C\mid Y)\le h(p)+p\log m\le p\log(e m/p),
 \label{eq:occupancy-entropy}
\end{equation}
where $h(p)=-p\log p-(1-p)\log(1-p)$ and the last expression is
defined as $0$ at $p=0$.
Also $x\mapsto x\log(e m/x)$ is increasing on $[0,1]$.
\end{lem}

\begin{proof}
Let $J=\ind_{\{C\ne0\}}$. The chain rule gives
\[
 H(C\mid Y)=H(J\mid Y)+H(C\mid J,Y)
       \le h(p)+p\log m.
\]
Here $H(J\mid Y)\le H(J)=h(p)$ and, given $J=1,Y$, at most $m$
labels remain. The inequality $-(1-p)\log(1-p)\le p$ gives the
second bound. Finally the derivative of $x\log(e m/x)$ is
$\log(m/x)\ge0$ for $0<x\le1$.
\end{proof}

The factor $p$ in this last estimate matters: uncertainty about
a coordinate is incurred only when an output is nonzero.
We will preserve this factor in the final accounting.

\begin{lem}[Independent inputs share an information budget]
\label{lem:product-information}
If $X_1,\ldots,X_D$ are independent conditional on $Y$, then
for any finite random variable $Z$,
\[
 \sum_{i=1}^D I(Z;X_i\mid Y)
 \le I(Z;X_1,\ldots,X_D\mid Y)
 =H(Z\mid Y)-H(Z\mid Y,X_1,\ldots,X_D).
\]
Additional private randomness in $Z$ is allowed.
\end{lem}
\begin{proof}
Let $X_{<i}=(X_1,\ldots,X_{i-1})$. Independence gives
$H(X_i\mid Y,X_{<i})=H(X_i\mid Y)$, while conditioning gives
$H(X_i\mid Z,Y,X_{<i})\le H(X_i\mid Z,Y)$.
Subtract to obtain
$I(Z;X_i\mid Y,X_{<i})\ge I(Z;X_i\mid Y)$.
Sum and apply the chain rule.
\end{proof}

\subsection{Table of Parameters and Constants}
\label{sec:occupancy-constants}

Table~\ref{tab:occupancy-constants} records the parameters and
numerical constants used in Sections~\ref{sec:model}--\ref{sec:main-proof}. These choices are deliberately
loose. The threshold $D_0$ is one absolute constant large enough
for all the stated parameter inequalities; it is independent
of the radius $r$ and of the selection rule. All logarithms in
this table and in the occupancy proof are natural.

\begingroup
\small
\setlength{\tabcolsep}{5pt}
\renewcommand{\arraystretch}{1.18}
\begin{longtable}{@{}>{\raggedright\arraybackslash}p{0.25\textwidth}
>{\raggedright\arraybackslash}p{0.46\textwidth}
>{\raggedright\arraybackslash}p{0.24\textwidth}@{}}
\caption{Parameters and numerical constants used in the occupancy proof.}
\label{tab:occupancy-constants}\\
\hline
Parameter or constant & Choice and purpose & Use or derivation \\
\hline
\endfirsthead
\multicolumn{3}{l}{\textit{Table~\ref{tab:occupancy-constants}, continued.}}\\
\hline
Parameter or constant & Choice and purpose & Use or derivation \\
\hline
\endhead
\hline
\endfoot
\multicolumn{3}{l}{\textit{Main parameters and list sizes}}\\
$D,k,r,D_0$ & Skeleton degree, clique size, view radius, and absolute degree threshold; $D\ge D_0$, $D\le\log k\le D^2$, $r\ge1$. & Theorem~\ref{thm:main}.\\
$\ell$ & $\lceil100\log(k+1)\rceil$: size of a list supplied by one shared edge view. & Lemma~\ref{lem:matrix}; \eqref{eq:list-sizes}.\\
$M$ & $\lceil1000\ell\rceil$: size of the list consolidating independent branch lists; $M\le D^3$. & \eqref{eq:list-sizes}.\\
$m$ & $DM$: number of nonzero coordinates retained before resampling; $m\le D^4$. & \eqref{eq:list-sizes}.\\
$q$ & Exceptional-matrix probability; $q\le e^{-\ell^2/10}\le D^{-2}$. & Lemma~\ref{lem:matrix}.\\
$\eta$ & $1/D=M/m$: an upper bound on the deleted mass captured by an $M$-coordinate list. & \eqref{eq:list-parameters}.\\
$\delta$ & $256\log D/D$: bound on the fraction of branches offering an unusually good short list. & Lemmas~\ref{lem:parameters} and \ref{lem:branch-transfer}.\\
$b$ & $2100\log D/D$: additive occupancy loss in one truncation. & Proposition~\ref{prop:truncation}.\\
\end{longtable}
\endgroup

\Needspace{5\baselineskip}
\section{Warmup: One-Round Hardness}
\label{sec:warmup}

This section proves that a one-round independent-set rule
cannot occupy a constant fraction of cliques as $D$ grows.
The proof has two steps: put almost all outputs in a fixed
short list, then show that an independent-set rule restricted
to that list rarely selects anything.

\subsection{1-Round Algorithms have Short Output Lists}

We begin with two key facts. The first rules out large
rectangles of compatible choices. We prove a density statement
because its quantitative form will also be useful when conflicts
have small positive probability.
\begin{lem}[Almost-empty random submatrices]
\label{lem:matrix}
Let $G$ be a $k$-by-$k$ matrix of independent Bernoulli-$1/2$
entries and let $\ell=\lceil100\log(k+1)\rceil\le k$.
Call $G$ exceptional if it contains an $\ell$-by-$\ell$
submatrix with at most $\ell^2/4$ entries equal to $1$.
Then
\begin{equation}
 q:=\Prob[G\text{ is exceptional}]
 \le k^{2\ell}e^{-\ell^2/8}
 \le e^{-\ell^2/10}.
 \label{eq:exceptional-matrix}
\end{equation}
\end{lem}

\begin{proof}
For $N$ independent Bernoulli-$1/2$ variables $B_j$,
\[
 \E\exp\left(-\lambda\sum_{j=1}^N(B_j-1/2)\right)
 =\cosh(\lambda/2)^N\le e^{N\lambda^2/8}.
\]
The last inequality follows from
$\log\cosh x\le x^2/2$; its derivative inequality is
$\tanh x\le x$ for $x\ge0$, and both functions are even.
Taking $\lambda=1$ in Markov's inequality yields
$\Prob[\sum_jB_j\le N/4]\le e^{-N/8}$.

Apply this with $N=\ell^2$ to a fixed submatrix and take a union
bound over at most $k^{2\ell}$ choices of its rows and columns.
Since $\log k\le\ell/100$, the resulting exponent is at most
$-\ell^2(1/8-2/100)\le-\ell^2/10$.
\end{proof}

The second fact explains how many independent, varying lists
can force membership in one fixed list. A label's \emph{weight}
is its expected number of occurrences among the random lists.
Only a few labels can have large weight because each list is short.
Independence makes it unlikely that a small-weight label
nevertheless appears many times.
\begin{lem}[Frequently occurring labels lie in a fixed list]
\label{lem:frequent-labels}
Let $S_1,\ldots,S_D$ be independent random subsets of $[k]$.
\footnote{The distribution that each $S_i$ is sampled from need not be uniform.}
For $a\in[k]$, put $w_a=\sum_i\Prob[a\in S_i]$ and
$C=\sum_aw_a=\sum_i\E|S_i|$.
Let $B$ consist of the $M$ labels of largest $w_a$, with ties
broken by coordinate order; if $M\ge k$, let $B=[k]$.
For any integers $M\ge1$ and $h\ge2$,
\begin{equation}
 \Prob[\text{some }a\notin B\text{ belongs to at least }h
                  \text{ of the }S_i]
 \le\frac{(eC/h)^h}{(h-1)M^{h-1}}.
 \label{eq:frequent-labels}
\end{equation}
\end{lem}

\begin{proof}
For a fixed $a$, the membership indicators in the $D$ lists
are independent. If $N_a=\sum_i\ind_{\{a\in S_i\}}$, then
\[
 \Prob[N_a\ge h]
 \le\E\binom{N_a}{h}
 \le\frac{w_a^h}{h!}
 \le\left(\frac{e w_a}{h}\right)^h.
\]
The middle inequality expands the factorial moment as a sum
over distinct indices and compares it to the full expansion
of $w_a^h$.
Order the weights decreasingly as $w_{(1)},\ldots,w_{(k)}$.
Since $w_{(j)}\le C/j$, a union bound gives
\[
 \Prob[\exists a\notin B:N_a\ge h]
 \le\left(\frac{eC}{h}\right)^h
       \sum_{j>M}j^{-h}
 \le\frac{(eC/h)^h}{(h-1)M^{h-1}},
\]
where the sum is bounded by the integral over $[M,\infty)$.
\end{proof}

We will repeatedly use the following elementary concentration
bound, with its proof included to fix the constants.

\begin{lem}[Too few eligible branches]
\label{lem:chernoff}
Let $L$ be a sum of independent indicator variables, with
$\nu=\E L$. For $0<a\le1$,
\[
 \Prob[L<(1-a)\nu]\le e^{-a^2\nu/2}.
\]
\end{lem}
\begin{proof}
If the indicators have success probabilities $p_i$, then
\[
 \E e^{-\lambda L}
 =\prod_i(1+p_i(e^{-\lambda}-1))
 \le \exp(\nu(e^{-\lambda}-1)).
\]
Apply Markov's inequality and use
$e^{-\lambda}\le1-\lambda+\lambda^2/2$ for $\lambda\ge0$.
Choosing $\lambda=a$ gives the displayed bound.
\end{proof}

Use the following list sizes for the remainder of the proof:
\begin{equation}
 \ell=\lceil100\log(k+1)\rceil,\qquad
 M=\lceil1000\ell\rceil,\qquad m=DM.
 \label{eq:list-sizes}
\end{equation}
Uniformly for $D\le\log k\le D^2$, and sufficiently large $D$,
\[
 \ell\le k,\quad M\le D^3,\quad m\le D^4,\quad
 q\le e^{-\ell^2/10}\le D^{-2},
\]
where $q$ is the exceptional-matrix probability.
Indeed $\ell=O(D^2)$, $M=O(D^2)$, $m=O(D^3)$, and
$k\ge e^D$. Every threshold here is independent of the rule.

\begin{prop}[A fixed list for a one-round independent-set rule]
\label{prop:warm-list}
Let $\mc A$ be a symmetric radius-one rule with
$\epsilon(\mc A)=0$. There is a fixed set $B\subseteq[k]$,
$|B|\le M$, such that
\[
 \Prob[\mc A_v\notin B\cup\{0\}]\le D^{-2}
\]
for all sufficiently large $D$ in the parameter regime above.
\end{prop}
\begin{proof}
For a fixed central matrix $g$, define
\[
 S_u(g)=\{a\in[k]:\Prob[\mc A_u=a\mid G_{uv}=g]>0\},
\]
and define $S_v(g)$ similarly, using the endpoint orientation.
Conditional on $g$, the outputs are independent by
Lemma~\ref{lem:views}. Their conditional conflict probability
is zero: it is nonnegative and its expectation is
$\epsilon(\mc A)=0$. Thus
$g(a,b)=0$ for every $a\in S_u(g),b\in S_v(g)$.
For a nonexceptional $g$, both sets cannot have size at least
$\ell$. It follows that
\[
 \ind_{\{|S_u(G_{uv})|<\ell\}}
 +\ind_{\{|S_v(G_{uv})|<\ell\}}
 \ge1
\]
unless $G_{uv}$ is exceptional.
Endpoint symmetry gives
$\Prob[|S_u(G_{uv})|<\ell]\ge(1-q)/2$.

At a root $v$, write $X_i=G_{vv_i}$ for its $D$ independent
incident matrices. Call branch $i$ \emph{eligible} when the
conditional nonzero support of $\mc A_v$ given $X_i$ has size
less than $\ell$. Let $S_i$ be that support on eligible values
and the empty set otherwise.
Each list depends only on $X_i$, so the lists and the eligibility
indicators are independent across $i$.
If $L$ counts eligible branches, then
$\nu=\E L\ge D(1-q)/2\ge D/3$ for large $D$.
Whenever $\mc A_v\ne0$, it belongs to every eligible list,
with probability one.

Lemma~\ref{lem:chernoff} gives
$\Prob[L<D/6]\le e^{-D/24}$.
Set $h=\lfloor D/6\rfloor\ge D/7$ for large $D$.
The total list weight is at most $D\ell$.
Let $B$ be the $M$ coordinates of largest weight.
Since $eD\ell/(hM)\le7e/1000<1/2$,
Lemma~\ref{lem:frequent-labels} gives
\[
 \Prob[\exists a\notin B:\ a\text{ occurs in at least }h
                 \text{ lists}]
 \le \frac{M}{h-1}2^{-h}.
\]
If the output is nonzero and outside $B$, then either $L<D/6$
or this last event occurs. Therefore its probability is at most
\[
 e^{-D/24}+\frac{M}{h-1}2^{-h}\le D^{-2}
\]
eventually, because $M\le D^3$ and $h\ge D/7$.
\end{proof}

\subsection{Short Lists Imply Low Occupancy}

We now turn a list-size bound into an occupancy bound.
This is the first use of hiding input information and resampling.
Here we hide all the input, so the new rule simply uses fixed
output probabilities.

\begin{lem}[Forgetting a one-round view]
\label{lem:warm-forget}
Let $\mc C$ be a symmetric radius-one independent-set rule
whose nonzero outputs lie in a fixed set of size at most $M$.
Write $z=p(\mc C)$, and let $\mc Q$ draw independently at
each vertex from the unconditional probabilities of $\mc C_v$.
Then
\[
 \frac{z^2}{2}=\epsilon(\mc Q)
 \le\frac{2H(\mc C_v)}D
 \le\frac{2z}{D}\log\frac{eM}{z}.
\]
\end{lem}
\begin{proof}
At $v$, the $D$ incident matrices $X_i$ are independent.
Lemma~\ref{lem:product-information} and branch symmetry give
$I(\mc C_v;X_i)\le H(\mc C_v)/D$ for a specified $i$.

Let $P$ be the actual distribution of
$(G_{uv},\mc C_u,\mc C_v)$. Note that $G_{uv}$ is one of the $X_i$s.
Let $Q$ have the same matrix distribution and independent
outputs with the unconditional probabilities of $\mc C$.
Lemma \ref{lem:resampling-kl} and the
conditional independence of the old outputs given $G_{uv}$ give
\[
 \begin{split}
 \KL(P\Vert Q)
 &=2H(\mc C_v)-H(\mc C_u,\mc C_v\mid G_{uv})\\
 &=2\bigl(H(\mc C_v)-H(\mc C_v\mid G_{uv})\bigr)
 \le \frac{2H(\mc C_v)}D.
 \end{split}
\]
The conflict event has probability zero under $P$.
Under $Q$, both outputs are nonzero with probability $z^2$,
and their selected matrix entry is an independent fair bit.
Thus it has probability $z^2/2$.
Apply Lemma~\ref{lem:sqrt-event} and square, then use
Lemma~\ref{lem:occupancy-entropy}.
\end{proof}

\begin{prop}[One-round occupancy bound]
\label{prop:warmup}
For sufficiently large $D$, uniformly for $D\le\log k\le D^2$,
every symmetric radius-one rule $\mc A$ with zero conflict has
\[
 p(\mc A)\le \frac{25\log D}{D}.
\]
\end{prop}
\begin{proof}
Use the fixed list $B$ from Proposition~\ref{prop:warm-list}.
Let $\mc C_v=\mc A_v$ when $\mc A_v\in B$, and let it be zero
otherwise. This only deletes selections, so $\epsilon(\mc C)=0$.
Its occupancy $z$ satisfies $p(\mc A)\le z+D^{-2}$.
If $z<1/D$, the conclusion follows immediately.
Otherwise Lemma~\ref{lem:warm-forget}, divided by $z>0$, gives
\[
 z\le\frac4D\log\frac{eM}{z}
 \le\frac{4(1+4\log D)}D
 \le\frac{20\log D}{D},
\]
where $M\le D^3$, $z\ge1/D$, and $\log D\ge1$.
Adding $D^{-2}$ proves the stated bound.
\end{proof}

The proof explains why the graph's random matrices do not give
the algorithm unlimited freedom to choose different coordinates.
Independence constrains the possible outputs in many separate
branches, and the information budget constrains coordination
even after those outputs are confined to one short list.

For larger radii, adjacent smaller views overlap and the
intermediate rules can have positive conflict probability.
We now extend the two parts of the warmup to address these
two issues, while keeping track of their quantitative costs.

\Needspace{5\baselineskip}
\section{Round Elimination by Conditional Resampling}
\label{sec:one-step}

Suppose a radius-$s$ output $Z_v$ depends on an inner view
$Y=N_{s-1}(v)$ and an outer input layer.
For every possible inner view $y$, compute the probabilities
\[
 q_a(y)=\Prob[Z_v=a\mid Y=y]\qquad(a=0,\ldots,k).
\]
These probabilities average over the hidden matrices and any
private randomness of the old rule. A new radius-$(s-1)$ rule
can output $a$ with probability $q_a(y)$.
It does not need to see, or reconstruct, the actual outer layer.
We call this operation \emph{conditional resampling}.
The proof allows arbitrary computation, so these finite
probabilities define a legitimate rule.

Conditional resampling preserves the probability of every
output at a single vertex after averaging over all inputs.
It can change how outputs at neighboring vertices depend on
each other and on the central matrix. That change is exactly
what must be bounded.

\begin{prop}[The information cost of removing one round]
\label{prop:one-step}
Let $\mc Z$ be a symmetric radius-$s$ rule, $s\ge1$.
Let $\mc Q$ be its conditional resampling onto $N_{s-1}(v)$.
Put
\[
 d=H(\mc Z_v\mid N_{s-1}(v))-H(\mc Z_v\mid N_s(v)).
\]
Then $d\ge0$, $p(\mc Q)=p(\mc Z)$, and
\begin{equation}
 \sqrt{\epsilon(\mc Q)}
 \le\sqrt{\epsilon(\mc Z)}+\sqrt{2d/D}.
 \label{eq:one-step}
\end{equation}
\end{prop}
\begin{proof}
Nonnegativity follows because conditioning decreases entropy,
and occupancy is preserved by averaging conditional probabilities.
Write $Y=N_{s-1}(v)$ and use the independent branch extensions
$X_1,\ldots,X_D$ from Lemma~\ref{lem:views}.
Their union with $Y$ is $N_s(v)$, so
Lemma~\ref{lem:product-information} gives
\[
 \sum_i I(\mc Z_v;X_i\mid Y)\le d.
\]
After averaging over $Y$, all $D$ terms are equal by symmetry.
For a specified incident edge, its shared information is
$F=(Y,X_i)$, and hence
\[
 H(\mc Z_v\mid Y)-H(\mc Z_v\mid F)\le d/D.
\]
We do not assert equality between branches after fixing one
particular inner view; the equality is in the averaged entropies.

Let $P$ be the actual joint distribution of
$(F,\mc Z_u,\mc Z_v)$. Let $Q$ keep the same $F$ distribution
and independently draw the two outputs from their respective
smaller-view probabilities. This is the distribution generated
by $\mc Q$ at the edge.
The old outputs are independent conditional on $F$ by
Lemma~\ref{lem:views}. Lemma~\ref{lem:resampling-kl} gives
\[
 \begin{split}
 \KL(P\Vert Q)
 &=2H(\mc Z_v\mid N_{s-1}(v))
                  -H(\mc Z_u,\mc Z_v\mid F)\\
 &=2\bigl(H(\mc Z_v\mid N_{s-1}(v))
                  -H(\mc Z_v\mid F)\bigr)\le 2d/D.
 \end{split}
\]
The tuple $F$ contains the central matrix, so conflict is the
same event in the two probability spaces. Apply
Lemma~\ref{lem:sqrt-event} to obtain the desired result.
\end{proof}

The quantity $H(\mc Z_v\mid N_s(v))$ is the uncertainty left
in the private output draw when all of the old input is known.
The number $d$ is the additional uncertainty caused by hiding
one layer. This interpretation will matter when we sum the
costs of successive reductions.

It would not help to apply this operation immediately to a rule
that might output any of the $k$ coordinates: its entropy could
be as large as $\log(k+1)$, which is too large in our regime.
As in the warmup, we first restrict outputs to a much shorter
list. At larger radii that list depends on the inner view.
To control the mass lost this way, we next need an estimate on
how strongly two neighboring inner views can coordinate their
deleted masses.

\Needspace{5\baselineskip}
\section{Multi-Round Entropy Bounds}
\label{sec:tree-entropy}

The independent-input calculation in the warmup is especially
simple because the root sees one independent matrix per branch.
At higher radius, the smaller views of adjacent vertices
overlap. We need a bound that respects this overlap.
The principle is that local outputs cannot share too much
information compared with their individual uncertainty.
We prove a conditional version, so that already known matrices
are accounted for exactly.

\subsection{Single-Layer Bounds}

We compare the entropy revealed by one additional layer at a
vertex with the entropy revealed by that layer at an edge.
The independent-input information budget applies twice:
first to the $D$ branches at a vertex, then to the two
independent completions of a shared edge view.

\begin{lem}[One-layer conditional entropy increment]
\label{lem:layer-entropy}
Let $\mc Z$ be a symmetric radius-$R$ rule. For adjacent $u,v$,
put
\[
 H_h=H(\mc Z_v\mid N_h(v)),\qquad
 J_h=H(\mc Z_u,\mc Z_v\mid F_h(uv)),
 \qquad 0\le h\le R.
\]
For every $0\le h<R$,
\[
 J_h-J_{h+1}\ge\frac{2(D-1)}D(H_h-H_{h+1}).
\]
\end{lem}
\begin{proof}
Fix $h<R$ and write $\Delta_h=H_h-H_{h+1}$.
At $v$, put $Y=N_h(v)$ and enumerate its neighbors as
$v_1,\ldots,v_D$ for analysis, with $v_1=u$.
Let $X_i$ be the matrix data in $F_h(vv_i)$ outside $Y$.
Lemma~\ref{lem:views}, applied with $s=h+1$, says that the
$X_i$ are independent conditional on $Y$, their union completes
$N_{h+1}(v)$, and $F_h(vv_i)$ is specified by $(Y,X_i)$.
Lemma~\ref{lem:product-information} therefore gives
\[
 \sum_{i=1}^D I(\mc Z_v;X_i\mid Y)
 \le I(\mc Z_v;X_1,\ldots,X_D\mid Y)
 =\Delta_h.
\]
The $D$ terms are equal by branch symmetry after averaging
over $Y$. Consequently
$I(\mc Z_v;X_1\mid Y)\le\Delta_h/D$.

Write $F=F_h(uv)$ and $A=(X_2,\ldots,X_D)$.
Since $F$ is specified by $(Y,X_1)$, the chain rule gives
\[
 I(\mc Z_v;A\mid F)
 =\Delta_h-I(\mc Z_v;X_1\mid Y)
 \ge\frac{D-1}{D}\Delta_h.
\]
Define $B$ in the same way at $u$, using all branches except
the one toward $v$. The corresponding inequality is
\[
 I(\mc Z_u;B\mid F)\ge\frac{D-1}{D}\Delta_h.
\]
The same $\Delta_h$ applies at both endpoints by symmetry.

The tuples $A$ and $B$ contain exactly the matrix data in
$N_{h+1}(v)$ and $N_{h+1}(u)$, respectively, outside $F$.
By Lemma~\ref{lem:views}, these are disjoint independent
completions conditional on $F$, and together they complete
$F_{h+1}(uv)$. Apply Lemma~\ref{lem:product-information}
to these two inputs and the joint output to obtain
\[
 \begin{aligned}
 J_h-J_{h+1}
 &=I((\mc Z_u,\mc Z_v);A,B\mid F)\\
 &\ge I((\mc Z_u,\mc Z_v);A\mid F)
      +I((\mc Z_u,\mc Z_v);B\mid F)\\
 &\ge I(\mc Z_v;A\mid F)+I(\mc Z_u;B\mid F)\\
 &\ge\frac{2(D-1)}D\Delta_h.
 \end{aligned}
\]
The second inequality retains one component of the joint
output in each mutual information term. The outputs
themselves need not be independent conditional on $F$.
\end{proof}

\subsection{Multiple Layers via Telescoping}

The one-layer comparison telescopes over all hidden layers.
At the complete input view, only the independent private
output draws remain uncertain.

\begin{lem}[Conditional edge entropy]
\label{lem:conditional-edge-entropy}
Let $\mc Z$ be a symmetric radius-$R$ rule, and let $0\le t\le R$.
At adjacent $u,v$, use
$F_t=\{G_{uv}\}\cup N_t(u)\cup N_t(v)$.
With $K=H(\mc Z_v\mid N_R(v))$, we have
\begin{equation}
 \frac D2 H(\mc Z_u,\mc Z_v\mid F_t)
 \ge (D-1)H(\mc Z_v\mid N_t(v))+K.
 \label{eq:conditional-edge-entropy}
\end{equation}
\end{lem}
\begin{proof}
Use $H_h,J_h$ from Lemma~\ref{lem:layer-entropy}.
Summing its inequality for $h=t,\ldots,R-1$ gives
\[
 J_t-J_R\ge\frac{2(D-1)}D(H_t-H_R).
\]
When $t=R$, the sum is empty and both differences are zero.

By definition $H_R=K$.
The shared view $F_R$ determines the complete input views
of both outputs. Conditional on $F_R$, their private draws
are independent, and the conditional distribution of each
output depends only on its own $N_R$ view. Hence
\[
 J_R=H(\mc Z_u\mid F_R)+H(\mc Z_v\mid F_R)=2K.
\]
Substitute these identities into the telescoped inequality:
\[
 J_t\ge 2K+\frac{2(D-1)}D(H_t-K)
      =\frac{2(D-1)}D H_t+\frac2D K.
\]
Multiplying by $D/2$ proves the claim.
\end{proof}

The estimate concerns conditional entropy of the outputs,
so it does not include the entropy of the surrounding raw
matrices. Its bound has no factor for the number of hidden
layers because the entropy increments telescope.
All conditional entropies and mutual informations above are
averaged over their conditioning variables. The input
independence used in the proof is conditional on raw views,
not on an output or on a selection event.

\subsection{Implications for Local Probabilities}

The next application turns an arbitrary local number in $[0,1]$
into the success probability of a private coin. This lets the
entropy inequality control how rarely large values can meet
across an edge.
\begin{lem}[Entropy bound for scalar coincidences]
\label{lem:scalar-entropy}
Let $f_v$ be the same symmetric finite-radius function of the
edge inputs at every vertex, with $0\le f_v\le1$ and $\mu=\E f_v$.
For adjacent $u,v$,
\begin{equation}
 \E[f_uf_v]\ge\frac{\mu^2}{2}-\frac{2h(\mu)}D.
 \label{eq:scalar-entropy}
\end{equation}
\end{lem}

\begin{proof}
Independently at each vertex conditional on the matrices, sample
$B_v\in\{0,1\}$ with success probability $f_v$.
Then $\Prob[B_v=1]=\mu$ and
$e:=\Prob[B_u=B_v=1]=\E[f_uf_v]$.
Lemma~\ref{lem:conditional-edge-entropy} with $t=0$ implies
\footnote{This uses the fact that $B$ is a rule in an unexpected way; i.e. that it only is allowed to output 0 or 1.}
\[
 H(B_u,B_v)\ge H(B_u,B_v\mid G_{uv})
              \ge 2(1-1/D)h(\mu).
\]
Thus $I(B_u;B_v)\le2h(\mu)/D$.
Compare their joint distribution $P$ (with $P(i)$ being the marginal for $i$) with the product of their marginals $Q(i,j) = P(i)P(j)$:
\[
\KL(P\Vert Q) = \sum_{i=0}^1 \sum_{j=0}^1 P(i,j)\log(P(i,j)/Q(i,j)) = 2\sum_{i=0}^1 P(i)\log (1/P(i)) - H(P) = I(B_u;B_v)
\]
Thus, for the event that both variables equal $1$,
Lemma~\ref{lem:sqrt-event} gives
$|\sqrt e-\mu|\le\sqrt{2h(\mu)/D}$.
Hence $\mu^2\le2e+4h(\mu)/D$, proving the assertion.
\end{proof}

The error term is $h(\mu)/D$, which is small when the mean
$\mu$ is small. This dependence on the mean is important when
the number $f_v$ represents a small probability of deletion.

\Needspace{5\baselineskip}
\section{Obtaining Shortlists via Round Elimination}
\label{sec:truncation}

In the one-round warmup, zero conflict forced one endpoint's
conditional support to be short. An intermediate rule obtained
by resampling can have positive conflict probability, and then
that exact conclusion need not hold.
We need a statement about the probability outside short lists.
We develop its matrix, branch, and neighboring-vertex parts
in that order.

\subsection{Conflict Probability Bounds Edge Tails}

For a typical central matrix, two independent output
distributions cannot both place substantial probability
outside every $\ell$-coordinate list without creating conflicts.
The following lemma remains valid when the two distributions
were themselves chosen after seeing the matrix.
\begin{lem}[Two diffuse distributions across one matrix]
\label{lem:diffuse}
Let $G$ be a nonexceptional matrix as in
Lemma~\ref{lem:matrix}, and let $P,Q$ be arbitrary distributions
on $\{0\}\cup[k]$. If $A\sim P$ and $B\sim Q$ independently and
\[
 e=\Prob[A,B\ne0,\ G(A,B)=1],
\]
then
\begin{equation}
 \tail_\ell(P)\tail_\ell(Q)\le4e.
 \label{eq:diffuse}
\end{equation}
The distributions may depend arbitrarily on $G$.
\end{lem}

\begin{proof}
Put $a=\tail_\ell(P)$ and $b=\tail_\ell(Q)$. If either is zero
the conclusion holds. Otherwise sample $\ell$ distinct nonzero
labels $A_1,\ldots,A_\ell$ by drawing sequentially from $P$,
conditioning each draw to avoid $0$ and all previous labels.
At every stage, the conditioning event has $P$-probability at
least $a$. Therefore, for every stage $i$ and every label $x$,
\[
 \Prob[A_i=x]\le P(x)/a.
\]
Independently perform the corresponding procedure for $Q$ to
obtain $B_1,\ldots,B_\ell$, with
$\Prob[B_j=y]\le Q(y)/b$.

For every $i,j$, independence between the two procedures gives
\[
 \E[G(A_i,B_j)]
 \le\frac{1}{ab}\sum_{x,y\in[k]}P(x)Q(y)G(x,y)
 =\frac{e}{ab}.
\]
If $e<ab/4$, the expected edge density of the sampled submatrix
would be less than $1/4$, so some such submatrix would have
density less than $1/4$. This contradicts nonexceptionality.
\end{proof}

\begin{prop}[Conditional endpoint tails]
\label{prop:endpoint-tails}
Let $\mc C$ be a symmetric radius-$s$ output rule, with $s\ge1$.
For an edge $uv$ let $F=F_{s-1}(uv)$ and define
\[
 z_u(F)=\tail_\ell(\mc C_u\mid F),
 \qquad z_v(F)=\tail_\ell(\mc C_v\mid F),
\]
where $\ell$ and $q$ are as in Lemma~\ref{lem:matrix}. Then
\begin{equation}
 \E[z_u(F)z_v(F)]\le4\epsilon(\mc C)+q.
 \label{eq:endpoint-tails}
\end{equation}
\end{prop}

\begin{proof}
By Lemma~\ref{lem:views}, the endpoint outputs are independent
conditional on $F$, and the central matrix is known from $F$.
For a nonexceptional central matrix,
Lemma~\ref{lem:diffuse} bounds $z_uz_v$ by four times the
conditional conflict probability. On the exceptional event use
$z_uz_v\le1$. Average over $F$ to obtain the claim.
In particular, no independence between the conditional output
distributions and the central matrix is being assumed.
\end{proof}

\subsection{Edge Tails Bound Vertex Tails}

We next generalize the independent-list argument from the
warmup. An eligible branch\footnote{Recall that a branch is an incident edge-neighborhood to the root.} now provides a list that captures
the output with probability at least $1-t/2$, rather than
with probability one. The capture events at different branches
need not be independent, because they concern the same output.
Only the eligibility indicators and the lists themselves
are independent across branches.

\begin{lem}[Many eligible lists capture the output together]
\label{lem:certificates}
Let $X_1,\ldots,X_D$ be independent finite random variables.
For each $i$, let $L_i$ be an event determined by $X_i$, and
let $C_i$ be an event on the full probability space.
Suppose $0<t\le1$ and
$\Prob[C_i\mid X_i=x]\ge1-t/2$ for every eligible value $x$
(that is, every $x$ for which $L_i$ occurs).
Put $\nu=\sum_i\Prob[L_i]$.
If $e^{-\nu/128}\le t/12$, then
\[
 \Prob\left[\sum_i\ind_{L_i\cap C_i}\ge\nu/8\right]\ge1-3t/4.
\]
\end{lem}
\begin{proof}
Let $L$ count eligible branches, let $T$ count eligible
branches whose capture event holds, and put $U=L-T$.
The hypothesis gives $\E[U]\le t\nu/2$.
Since the eligibility indicators are independent,
Lemma~\ref{lem:chernoff} with $a=1/8$ gives
$\Prob[L<7\nu/8]\le e^{-\nu/128}$.
On $\{T<\nu/8,L\ge7\nu/8\}$ we have $U\ge3\nu/4$.
Markov's inequality therefore yields
\[
 \Prob[T<\nu/8]
 \le\frac{\E[U]}{3\nu/4}+e^{-\nu/128}
 \le\frac{2t}{3}+\frac{t}{12}=\frac{3t}{4}.
\]
\end{proof}

The occurrence threshold is a constant fraction $\nu/8$,
even when $t$ is small. The small value of $t$ is used to
bound the probability that so many capture events fail.

\begin{lem}[Parameter bounds for branch consolidation]
\label{lem:parameters}
Use $\ell,M,m$ from \eqref{eq:list-sizes}, and set
\begin{equation}
 \eta=\frac1D,\qquad \delta=\frac{256\log D}{D}.
 \label{eq:list-parameters}
\end{equation}
For sufficiently large $D$, uniformly for $D\le\log k\le D^2$,
we have $M/m=\eta$, $m\le D^4$, $\delta<1$, $q\le D^{-2}$,
and, for every integer $h\ge32\log D-1$,
\[
 h\ge2,\qquad \frac{M}{h-1}2^{-h}<\frac{1}{8D}.
\]
\end{lem}
\begin{proof}
The size and matrix bounds were verified after
\eqref{eq:list-sizes}. The ratio $M/m$ is exact and
$256\log D/D<1$ eventually.
For the last inequality, use $M\le D^3$ and monotonicity in $h$:
\[
 \frac{M}{h-1}2^{-h}
 \le\frac{D^3\,2^{-(32\log D-1)}}{32\log D-2}
 =o(D^{-1}).
\]
All thresholds depend only on absolute constants.
\end{proof}

The next lemma says that if every fixed $M$-coordinate list
leaves probability at least $t$ outside, then only a small
fraction of branches can offer an $\ell$-coordinate list
with the substantially smaller error $t/2$.
\begin{lem}[Relative branch transfer with a uniform exception bound]
\label{lem:branch-transfer}
Use the parameters of Lemma~\ref{lem:parameters}.
Let $X_1,\ldots,X_D$ be independent finite random variables, and
let $Z\in\{0\}\cup[k]$ have any conditional distribution given them.
Put $t=\tail_M(Z)$.
If $t\ge1/D$, then, for an independent uniform $I\in[D]$,
\begin{equation}
 \Prob_{I,X_I}\left[\tail_\ell(Z\mid I,X_I)<t/2\right]\le\delta.
 \label{eq:branch-transfer}
\end{equation}
The assertion also holds conditional on any fixed inner view
under which the branch extensions are independent.
\end{lem}

\begin{proof}
Suppose the displayed probability exceeds $\delta$.
Call a value of $X_i$ eligible when its conditional $\ell$-tail
is less than $t/2$. Choose a minimizing $\ell$-label list there,
and use the empty list otherwise.
The resulting random lists $S_i$ are independent.
The mean eligible count satisfies $\nu>\delta D=256\log D$.
At eligible values the capture event $Z\in S_i\cup\{0\}$
has conditional probability greater than $1-t/2$.
For sufficiently large $D$,
$e^{-\nu/128}<D^{-2}\le t/12$.
Lemma~\ref{lem:certificates} gives simultaneous capture in at least
$\nu/8$ eligible lists with probability at least $1-3t/4$.

Set $h=\lfloor\nu/8\rfloor$.
Then $h\ge32\log D-1$ and $h\ge\nu/16$.
The total list weight in Lemma~\ref{lem:frequent-labels} is
$C\le\ell\nu$, so $eC/(hM)\le16e/1000<1/2$.
The $M$ labels of largest weight form a fixed list $B$.
By Lemmas~\ref{lem:frequent-labels} and \ref{lem:parameters}, the
probability that some label outside $B$ occurs in at least $h$
lists is less than $1/(8D)\le t/8$.
Thus
\[
 \Prob[Z\in B\cup\{0\}]>1-3t/4-t/8=1-7t/8>1-t,
\]
contradicting the definition of $t$.
For a fixed inner view, apply the same argument to the conditional
product distribution. It uses no identical-distribution assumption.
\end{proof}

\subsection{Sandwiching Vertex Tails between Occupancy and Conflict via Entropy}

Fix a radius-$s$ rule and a smaller view $Y=N_{s-1}(v)$.
Retain its $m$ most probable nonzero coordinates given $Y$.
The \emph{kept rule} returns the original coordinate when it
belongs to this list and returns $0$ otherwise.
The \emph{deleted rule} returns the original coordinate exactly
when it is nonzero and outside this list.
They partition the original nonzero outputs at each vertex.

Why should the deleted mass be small?
Every deleted coordinate has conditional probability at most
$1/m$, so a much shorter list of $M=m/D$ coordinates can
capture at most $1/D$ of that mass.
The branch lemma (Lemma \ref{lem:branch-transfer})
says that this spread in the smaller view
usually persists in the shared edge view.
The endpoint-tails proposition (Proposition \ref{prop:endpoint-tails})
then forces a conflict when it persists at
both endpoints. The scalar entropy estimate (Lemma \ref{lem:scalar-entropy}) ensures that
the large deleted masses cannot avoid meeting too often.
\begin{prop}[Conditional truncation]
\label{prop:truncation}
Use parameters from Lemma~\ref{lem:parameters}.
For a symmetric radius-$s$ rule $\mc A$, $s\ge1$, let $P_Y$
be its $m$ most probable nonzero labels given $Y=N_{s-1}(v)$.
Let $\mc A^{\keep}$ retain exactly the original outputs in $P_Y$,
and let $\mc A^{\del}$ retain exactly the original nonzero outputs
outside $P_Y$. Other outputs become zero.
With $\tau=p(\mc A^{\del})$ and $w=\epsilon(\mc A^{\del})$,
\begin{align}
 p(\mc A^{\keep})&=p(\mc A)-\tau, \label{eq:kept-marginal}\\
 \epsilon(\mc A^{\keep})+w&\le\epsilon(\mc A),
       \label{eq:kept-conflict}\\
 \tau&\le\frac{2100\log D}{D}+6\sqrt w.
       \label{eq:truncation}
\end{align}
Both rules preserve symmetry and radius, and the kept rule is conditionally supported on $P_Y\cup\{0\}$.
\end{prop}

\begin{proof}
Break all ties by coordinate order.
The symmetry, support, and occupancy claims follow directly.
Couple the rules using the original private draw at each vertex.
The conflicts with two kept outputs and those with two deleted
outputs are disjoint subsets of the original conflict event.
The omitted mixed conflicts are nonnegative, proving
\eqref{eq:kept-conflict}.

Let $\tau_Y=\Prob[\mc A_v^{\del}\ne0\mid Y]$.
Every deleted label has conditional probability at most $1/m$.
Consequently
\[
 t_Y:=\tail_M(\mc A_v^{\del}\mid Y)\ge\tau_Y-M/m=\tau_Y-\eta.
\]
Put $f_v=\tfrac12(\tau_{Y_v}-2\eta)_+$ and $\mu=\E f_v$.
Then $0\le f_v\le1$ and $\tau\le2\mu+2\eta$.
At an edge $uv$ with shared view $F=F_{s-1}(uv)$ put
$z_{uv}=\tail_\ell(\mc A_u^{\del}\mid F)$ and $E_{uv}=\{z_{uv}<f_u\}$,
and define the corresponding quantities at $v$.

If $f_u=0$, failure is impossible.
Otherwise $t_{Y_u}>\eta$ and $f_u\le t_{Y_u}/2$, so
Lemma~\ref{lem:branch-transfer} bounds the conditional failure
probability, averaged over directions, by $\delta$.
Weight by the $Y_u$-measurable $f_u$ and then average over $Y_u$.
Unconditional branch symmetry gives, for a specified edge,
\[
 \E[\ind_{E_{uv}}f_uf_v]\le\E[\ind_{E_{uv}}f_u]=\frac{1}{D}\sum_{w\sim u}\E[\ind_{E_{uw}}f_u] = \frac{1}{D}\sum_{w\sim u}\E[f_u \Pr[E_{uw}\mid Y_u]]\le\delta\E[f_u]=\delta\mu.
\]
The same holds at $v$.
When neither $E_{uv}$ nor $E_{vu}$ happens, $z_{uv}z_{vu}\ge f_uf_v$.
Proposition~\ref{prop:endpoint-tails} and
Lemma~\ref{lem:scalar-entropy} imply
\[
 \frac{\mu^2}{2}-\frac{2h(\mu)}D
 \le\E[f_uf_v]\le4w+q+2\delta\mu.
\]

If $\mu<1/D$, then $\tau<4/D$, which suffices.
Otherwise
$h(\mu)\le\mu\log(e/\mu)\le2\mu\log D$ for large $D$.
Therefore
\[
 \mu^2\le8w+2q+\frac{1032\log D}{D}\mu,
 \qquad
 \mu\le\frac{1032\log D}{D}+2\sqrt2\,\sqrt w+\sqrt{2q}.
\]
Using $q\le D^{-2}$ and $\eta=D^{-1}$ gives
\[
 \tau\le\frac{2064\log D}{D}+4\sqrt2\,\sqrt w
                           +\frac{2\sqrt2+2}{D}
       \le\frac{2100\log D}{D}+6\sqrt w.
\]
\end{proof}

In this proof the two endpoint functions $f_u,f_v$ were never
assumed independent. The entropy estimate supplied precisely
the lower bound on $\E[f_uf_v]$ that was needed.
Likewise, branch symmetry was used only after averaging over
the inner view.

There is also an exact version of the conflict bookkeeping.
Couple kept and deleted outputs using the same original
private draw at each vertex. If $x$ is the probability of a
conflict having one kept and one deleted endpoint, counting
both orientations, then
\begin{equation}
 \epsilon(\mc A)
 =\epsilon(\mc A^{\keep})+\epsilon(\mc A^{\del})+x.
 \label{eq:mixed}
\end{equation}
In particular, when $\epsilon(\mc A)=0$, both restricted
rules have zero conflict. Conflicts can first appear when
the kept rule is conditionally resampled, not when it is
truncated.

\Needspace{5\baselineskip}
\section{Warmups to the Final Bound}
\label{sec:sequential}

We now have the two operations needed for iteration:
truncation loses little mass unless the deleted part creates
conflicts, and conditional resampling removes one round at a
controlled information cost.
Before using the final accounting, we analyze this iteration
in two direct ways.

\subsection{Consolidating One-Round Bounds}

Start with $\mc A_0=\mc A$, of radius $r$, and put
$V_j(v)=N_{r-j}(v)$ for $j=0,\ldots,r$.
At step $j<r$, let $P_j(V_{j+1}(v))$ be the $m$ most probable
nonzero outputs of $\mc A_j$ given the smaller view.
Let $\mc C_j$ and $\mc D_j$ be its kept and deleted rules.
Obtain $\mc A_{j+1}$ by conditionally resampling $\mc C_j$
onto $V_{j+1}$, using new private draws.
All these rules retain the required symmetry.

Define the numbers
\[
 \begin{gathered}
 p_j=p(\mc A_j),\quad
 \epsilon_j=\epsilon(\mc A_j),\quad
 \tau_j=p(\mc D_j)=p_j-p_{j+1},\quad
 w_j=\epsilon(\mc D_j),\\
 d_j=H((\mc C_j)_v\mid V_{j+1}(v))
              -H((\mc C_j)_v\mid V_j(v))\ge0,\qquad
 b=\frac{2100\log D}{D}.
 \end{gathered}
\]
We suppress $v$ in conditioning notation when no confusion is
possible. Propositions~\ref{prop:one-step} and
\ref{prop:truncation} give
\begin{equation}
 \tau_j\le b+6\sqrt{w_j},\qquad
 \sqrt{\epsilon_{j+1}}
 \le\sqrt{\epsilon_j-w_j}+\sqrt{2d_j/D}.
 \label{eq:iteration}
\end{equation}
The second inequality uses
$\epsilon(\mc C_j)\le\epsilon_j-w_j$ from
\eqref{eq:mixed}. Thus $\epsilon_j-w_j\ge0$.

The final rule $\mc A_r$ has radius zero.
\begin{lem}[Zero-radius conflict]
\label{lem:zero}
For a symmetric radius-zero rule, $\epsilon=p^2/2$.
\end{lem}

\begin{proof}
All vertices use the same fixed output distribution.
Their private draws are independent of each other and of the
central matrix. Both are nonzero with probability $p^2$,
and the selected entry is $1$ with probability $1/2$.
\end{proof}

\begin{lem}[Occupancy from a total conflict budget]
\label{lem:mass-from-budget}
For the above sequence, put
$S=\epsilon_r+\sum_{j<r}w_j$ and $p=p_0$.
Then
\begin{equation}
 p\le rb+\sqrt{36r+2}\sqrt S.
 \label{eq:mass-from-budget}
\end{equation}
\end{lem}
\begin{proof}
The occupancy differences telescope:
$p=p_r+\sum_{j<r}\tau_j$.
Use \eqref{eq:iteration} and Lemma~\ref{lem:zero} to obtain
\[
 p\le rb+6\sum_{j<r}\sqrt{w_j}+\sqrt{2\epsilon_r}.
\]
Cauchy--Schwarz on these $r+1$ terms proves the result.
\end{proof}

\begin{lem}[The cost of accumulating one-step errors]
\label{lem:sequential-budget}
If $\epsilon_0=0$, then
\begin{equation}
 S\le \frac2D\left(\sum_{j<r}\sqrt{d_j}\right)^2
      \le\frac{2r}{D}\sum_{j<r}d_j.
 \label{eq:sequential-budget}
\end{equation}
\end{lem}
\begin{proof}
Put $W_j=\sum_{i<j}w_i$ and $S_j=\epsilon_j+W_j$.
By \eqref{eq:iteration},
\[
 \begin{split}
 S_{j+1}
 &\le \epsilon_j-w_j+\frac{2d_j}{D}
       +2\sqrt{(\epsilon_j-w_j)\frac{2d_j}{D}}
       +W_j+w_j\\
 &\le S_j+\frac{2d_j}{D}
       +2\sqrt{S_j\frac{2d_j}{D}}
   =\left(\sqrt{S_j}+\sqrt{2d_j/D}\right)^2.
 \end{split}
\]
Since $S_0=0$, induction yields the first inequality.
Cauchy--Schwarz yields the second.
\end{proof}

\subsection{First ($\tilde{\Omega}(D^{1/3})$) Bound}

The kept rule has at most $m$ nonzero coordinates given
$V_{j+1}$. Its occupancy is $p_{j+1}\le p$.
Lemma~\ref{lem:occupancy-entropy} and monotonicity give
\[
 d_j\le H((\mc C_j)_v\mid V_{j+1})
       \le p\log(em/p).
\]
Using this same upper bound at all $r$ steps is legitimate,
but it charges a potentially surviving output repeatedly.

\begin{prop}[A first bound with a cubic term]
\label{prop:cubic}
Under the assumptions of Theorem~\ref{thm:main},
\[
 p\le \frac{4200r\log D}{D}
          +\frac{532r^3\log D}{D}.
\]
\end{prop}
\begin{proof}
The preceding estimate and Lemma~\ref{lem:sequential-budget}
give $S\le2r^2p\log(em/p)/D$.
Insert this into Lemma~\ref{lem:mass-from-budget}:
\[
 p\le rb+\sqrt{\frac{76r^3p}{D}\log\frac{em}{p}},
\]
where $2(36r+2)\le76r$.
If $r>D$ or $p<1/D$, the claimed bound is immediate.
Otherwise $m\le D^4$ and $p\ge1/D$ imply
$\log(em/p)\le7\log D$.
Use $\sqrt{Ap}\le(p+A)/2$ with
$A=532r^3\log D/D$, then rearrange.
\end{proof}

This is already a growing-round result: it gives $p=o(1)$
for $r=o((D/\log D)^{1/3})$.
There are two distinct losses in its accounting.
First, we bounded the sum of information costs by $r$ times
one worst-case cost. Second, we accumulated square-root errors
along the whole chain, producing the factor $r$ in
\eqref{eq:sequential-budget}.
We address the first loss next.

\subsection{Second ($\tilde{\Omega}(D^{1/2})$) Bound}

Conditional resampling changes uncertainty about unseen inputs
into private output randomness. This randomness is still
present at the next step; it has not disappeared.
Truncation can erase some of it by changing coordinates to $0$,
but only coordinates actually deleted can cause this loss.
The following calculation formalizes this explanation.

\begin{lem}[Entropy accounting over the entire sequence]
\label{lem:entropy-accounting}
For the reduction sequence,
\begin{equation}
 \sum_{j<r}d_j
 \le p_1\log\frac{emr}{p_1}
 \le p_0\log\frac{emr}{p_0},
 \label{eq:entropy-accounting}
\end{equation}
with zero-mass expressions interpreted as zero.
\end{lem}
\begin{proof}
Write $\Phi_j=H((\mc A_j)_v\mid V_j)$ for the private
output entropy at stage $j$.
For $j\ge1$, the outputs of $\mc A_j$ given $V_j$ lie in
$P_{j-1}(V_j)\cup\{0\}$, a set of at most $m$ nonzero
coordinates. The same is true for $\mc D_j$.
Set
\[
 L_j=\Phi_j-H((\mc C_j)_v\mid V_j)\quad(j\ge1).
\]
The kept output is a function of $\mc A_j$ and $V_j$, while
the pair of kept and deleted outputs determines $\mc A_j$.
The chain rule and conditioning therefore give
\[
 0\le L_j
 =H((\mc A_j)_v\mid(\mc C_j)_v,V_j)
 \le H((\mc D_j)_v\mid V_j)
 \le \tau_j\log(em/\tau_j).
\]
By the definition of conditional resampling,
\[
 \Phi_{j+1}=H((\mc C_j)_v\mid V_{j+1})
           =H((\mc C_j)_v\mid V_j)+d_j.
\]
At $j=0$ this reads
$\Phi_1=H((\mc C_0)_v\mid V_0)+d_0$.
At $j\ge1$ it reads $\Phi_{j+1}=\Phi_j-L_j+d_j$.
Telescoping these equalities gives
\[
 \sum_{j<r}d_j
 =\Phi_r-H((\mc C_0)_v\mid V_0)+\sum_{j=1}^{r-1}L_j
 \le p_r\log(em/p_r)+\sum_{j=1}^{r-1}\tau_j\log(em/\tau_j).
\]
For the last step, $V_r$ is empty and the final output has at
most $m$ nonzero coordinates.

The $r$ masses $p_r,\tau_1,\ldots,\tau_{r-1}$ sum to $p_1$.
If $p_1>0$, normalize them to a probability vector
$\alpha_1,\ldots,\alpha_r$.
Then the right side equals
\[
 p_1\log(em/p_1)+p_1\sum_{i=1}^r\alpha_i\log(1/\alpha_i)
 \le p_1\log(emr/p_1).
\]
The bound on the entropy of an $r$-valued variable was proved
in Section~\ref{sec:model}.
If $p_1=0$, all terms vanish.
Finally $p_1\le p_0$ and $x\log(emr/x)$ is increasing on
$[0,1]$.
\end{proof}

\begin{prop}[A second bound with a quadratic term]
\label{prop:quadratic}
Under the assumptions of Theorem~\ref{thm:main},
\[
 p\le\frac{4200r\log D}{D}
         +\frac{532r^2\log D}{D}.
\]
\end{prop}
\begin{proof}
Combine Lemmas~\ref{lem:sequential-budget} and
\ref{lem:entropy-accounting} to get
$S\le2rp\log(emr/p)/D$.
Lemma~\ref{lem:mass-from-budget} gives
\[
 p\le rb+\sqrt{\frac{76r^2p}{D}\log\frac{emr}{p}}.
\]
If $r>D$ or $p<1/D$, the claim is immediate.
Otherwise $m\le D^4$, $r\le D$, and $p\ge1/D$ imply
$\log(emr/p)\le1+6\log D\le7\log D$.
Apply $\sqrt{Ap}\le(p+A)/2$ with
$A=532r^2\log D/D$.
\end{proof}

This gives $p=o(1)$ for $r=o(\sqrt{D/\log D})$.
The only remaining extra factor in this accounting is the
$r$ in the conversion from $\sum_jd_j$ to $S$.
It comes from following the error of an intermediate rule
through every later resampling.
We now stop following that chain and instead compare each
deleted part directly with a restriction of the original
zero-conflict rule.

\Needspace{5\baselineskip}
\section{Exploiting Multi-Round Entropy Bounds}
\label{sec:joint}

There are two tasks left. First, show that hiding several
layers at once has a cost depending only on the uncertainty
remaining in the smaller view. Second, identify an appropriate
zero-conflict original rule to which this estimate can be
applied for each deletion.

\subsection{A Multi-Step Version of Proposition \ref{prop:one-step}}

The one-step proof in Section~\ref{sec:one-step} used
independence of the old endpoint outputs after fixing their
shared view. If many layers are hidden, that independence
need not hold. The conditional edge-entropy inequality from
Section~\ref{sec:tree-entropy} replaces it.

\begin{prop}[Direct comparison with a smaller-view rule]
\label{prop:direct-projection}
Let $\mc Z$ be a symmetric radius-$R$ rule, and let
$0\le t\le R$. Define the radius-$t$ rule $\mc Q$ by
independent private draws with probabilities
\[
 \Prob[\mc Q_v=a\mid N_t(v)=y]
       =\Prob[\mc Z_v=a\mid N_t(v)=y]
       \qquad(a=0,\ldots,k).
\]
Then $p(\mc Q)=p(\mc Z)$ and
\begin{equation}
 \sqrt{\epsilon(\mc Q)}
 \le\sqrt{\epsilon(\mc Z)}
 +\sqrt{\frac2D\left(
       H(\mc Z_v\mid N_t(v))-H(\mc Z_v\mid N_R(v))
                       \right)}.
 \label{eq:direct-projection}
\end{equation}
In particular, if $\epsilon(\mc Z)=0$, $z := p(\mc Z)$,
and at most $m$ nonzero coordinates are possible given each
view $N_t(v)$, then
\begin{equation}
 \epsilon(\mc Q)\le\frac{2z}{D}\log\frac{em}{z}.
 \label{eq:independent-projection}
\end{equation}
The last expression is zero at $z=0$.
\end{prop}
\begin{proof}
Occupancy is preserved by the law of total probability.
Let $P$ be the actual distribution of
$(F_t,\mc Z_u,\mc Z_v)$, with
$F_t=\{G_{uv}\}\cup N_t(u)\cup N_t(v)$.
Let $Q$ keep the same $F_t$ distribution and draw the
two outputs independently from their smaller-view
probabilities. Because both smaller views are determined
by $F_t$, this is precisely the edge distribution of $\mc Q$.
By the first equality of Lemma~\ref{lem:resampling-kl},
\[
 \KL(P\Vert Q)
 =2H(\mc Z_v\mid N_t(v))
                -H(\mc Z_u,\mc Z_v\mid F_t).
\]
Unlike our other applications of Lemma~\ref{lem:resampling-kl},
this identity does not require conditional independence of
the old outputs. Lemma~\ref{lem:conditional-edge-entropy}
bounds it by
\[
 \frac2D\bigl(H(\mc Z_v\mid N_t(v))
                         -H(\mc Z_v\mid N_R(v))\bigr).
\]
Apply Lemma~\ref{lem:sqrt-event} to the conflict event.
The central matrix is part of $F_t$, so the event is
identical under the two distributions.
This proves \eqref{eq:direct-projection}.
When the old conflict is zero, square the inequality,
drop the nonnegative full-view entropy, and apply
Lemma~\ref{lem:occupancy-entropy}.
\end{proof}

There is no factor $R-t$ in this bound.
It is useful when a rule with zero conflict has few possible
nonzero coordinates after conditioning on the smaller view.
The original rule need not have that property, so we next
construct restrictions of it that do.

\subsection{Combining Truncation Steps}

Return to the sequence from Section~\ref{sec:sequential}.
Once $\mc A_0$ is fixed, every set-valued function
$P_j(V_{j+1})$ is fixed as well: it is computed from the
conditional output probabilities, not from a realized private
draw. We can therefore test the \emph{same original output}
against all these lists.

At a vertex where $\mc A_0$ outputs a nonzero coordinate $a$,
ask whether $a$ belongs to $P_0(V_1)$, then to $P_1(V_2)$,
and so on. Record the first rejection, or record that $a$
passes every list. Keep the coordinate unchanged whenever
it is recorded.

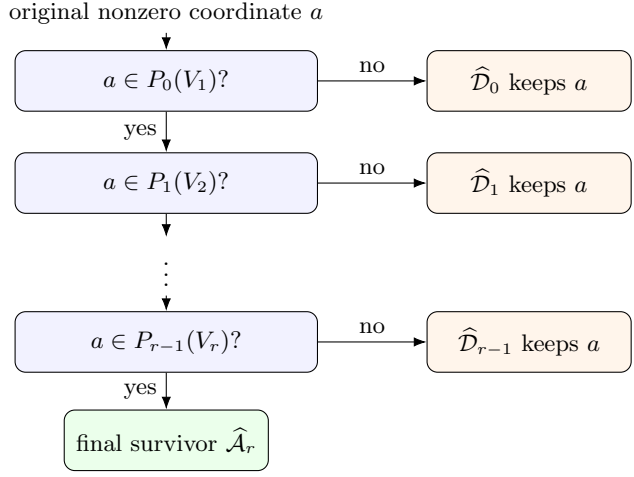
\begin{figure}[htbp]
\centering
\begin{tikzpicture}[x=1cm,y=1cm,>=Latex,
 test/.style={draw,rounded corners,fill=blue!5,align=center,
              minimum width=4cm,minimum height=8mm,font=\small},
 resultbox/.style={draw,rounded corners,fill=orange!8,align=center,
              minimum width=2.7cm,minimum height=8mm,font=\small}]
 \node[font=\small] (start) at (0,0.55) {original nonzero coordinate $a$};
 \node[test] (t0) at (0,-0.35) {$a\in P_0(V_1)$?};
 \node[test] (t1) at (0,-1.7) {$a\in P_1(V_2)$?};
 \node[font=\small] (dots) at (0,-2.8) {$\vdots$};
 \node[test] (tr) at (0,-3.8) {$a\in P_{r-1}(V_r)$?};
 \node[resultbox,fill=green!8] (last) at (0,-5.1) {final survivor $\widehat{\mc A}_r$};
 \node[resultbox] (d0) at (4.8,-0.35) {$\widehat{\mc D}_0$ keeps $a$};
 \node[resultbox] (d1) at (4.8,-1.7) {$\widehat{\mc D}_1$ keeps $a$};
 \node[resultbox] (dr) at (4.8,-3.8) {$\widehat{\mc D}_{r-1}$ keeps $a$};
 \draw[->] (start)--(t0);
 \draw[->] (t0)--node[left,font=\small] {yes}(t1);
 \draw[->] (t1)--(dots);
 \draw[->] (dots)--(tr);
 \draw[->] (tr)--node[left,font=\small] {yes}(last);
 \draw[->] (t0)--node[above,font=\small] {no}(d0);
 \draw[->] (t1)--node[above,font=\small] {no}(d1);
 \draw[->] (tr)--node[above,font=\small] {no}(dr);
\end{tikzpicture}
\caption{Every original nonzero output belongs to exactly one
first-failure part or to the final survivor. The coordinate
$a$ is unchanged throughout these tests. There are no new
output draws in this diagram. Each resulting part is a
restriction of the original independent-set selection.}
\label{fig:first-failure}
\end{figure}

\begin{lem}[Original survivors have the required conditional probabilities]
\label{lem:ancestors}
Use the lists and views of Section~\ref{sec:sequential}.
On the probability space of the original rule, define
\[
 (\widehat{\mc A}_j)_v=
 \begin{cases}
  (\mc A_0)_v,&
       (\mc A_0)_v\ne0\text{ and }
       (\mc A_0)_v\in P_i(V_{i+1}(v))\text{ for every }i<j,\\
  0,&\text{otherwise}.
 \end{cases}
\]
Let $\widehat{\mc D}_j$ retain the original nonzero coordinate
exactly when it passes every earlier list and fails
$P_j(V_{j+1})$. Then, for every $a\in\{0,\ldots,k\}$,
\begin{align}
 \Prob[(\mc A_j)_v=a\mid V_j]
    &=\Prob[(\widehat{\mc A}_j)_v=a\mid V_j],
      \label{eq:ancestor-law}\\
 \Prob[(\mc D_j)_v=a\mid V_j]
    &=\Prob[(\widehat{\mc D}_j)_v=a\mid V_j].
      \label{eq:deleted-ancestor-law}
\end{align}
All hatted rules are symmetric radius-$r$ restrictions of
$\mc A_0$. If $\epsilon(\mc A_0)=0$, they all have zero conflict.
Their first-failure parts and the final survivor partition
the original nonzero outputs.
For $j\ge1$, the possible nonzero outputs of
$\widehat{\mc A}_j$ and $\widehat{\mc D}_j$, given $V_j$,
belong to a list of at most $m$ coordinates.
\end{lem}
\begin{proof}
Prove \eqref{eq:ancestor-law} for nonzero $a$ by induction.
It is immediate at $j=0$.
The next list depends only on $V_{j+1}\subset V_j$.
By the definition of the kept rule and of conditional
resampling,
\[
 \begin{split}
 \Prob[(\mc A_{j+1})_v=a\mid V_{j+1}]
 &=\ind_{\{a\in P_j(V_{j+1})\}}\,
    \E\!\left[\Prob[(\mc A_j)_v=a\mid V_j]\mid V_{j+1}\right]\\
 &=\ind_{\{a\in P_j(V_{j+1})\}}\,
    \Prob[(\widehat{\mc A}_j)_v=a\mid V_{j+1}]\\
 &=\Prob[(\widehat{\mc A}_{j+1})_v=a\mid V_{j+1}].
 \end{split}
\]
The middle equality is the induction hypothesis followed
by averaging over the larger view. This is the tower property
for conditional probabilities.
The probability of $0$ is the remaining probability.
Multiplying the stage-$j$ identity by
$\ind_{\{a\notin P_j(V_{j+1})\}}$ gives
\eqref{eq:deleted-ancestor-law} for nonzero $a$, and hence
also for zero.

Each hatted output is either the original coordinate or zero,
chosen by tests on radius-$r$ matrix data. The tests are
symmetric. Conditional on all matrices, the outputs at
different vertices still depend on independent original
private draws. Thus these are rules in our model.
Their conflict events are subsets of the original conflict
event. Also, an original nonzero coordinate has a unique
first failed test unless it passes them all.

Finally, for $j\ge1$, every nonzero output of either hatted
rule has passed $P_{j-1}(V_j)$. Once $V_j$ is fixed, this
is a fixed list of size at most $m$.
\end{proof}

The lemma asserts equality of conditional probabilities,
not equality between the realized outputs of a resampled
rule and the original rule.
That distinction is necessary: the intermediate rules make
new private draws. It is enough because a rule in our model
is specified by exactly these smaller-view probabilities.

\subsection{Improved Bound on $S$}

\begin{prop}[Joint deleted-conflict budget]
\label{prop:joint-budget}
Assume
$\epsilon(\mc A_0)=0$. Recall that
$p_j=p(\mc A_j)$, $\tau_j=p_j-p_{j+1}$,
$w_j=\epsilon(\mc D_j)$, and $\epsilon_r=\epsilon(\mc A_r)$.
Then
\begin{equation}
 \boxed{\quad
 \epsilon_r+\sum_{j<r}w_j
 \le \frac{2p_1\log(e m r/p_1)}D
 \le \frac{2p_0\log(e m r/p_0)}D.
 \quad} \label{eq:joint-budget}
\end{equation}
Zero-mass expressions (when $p_1=0$ or $p_0=0$) are interpreted by continuity.
\end{prop}

\begin{proof}
Use the construction of Lemma~\ref{lem:ancestors}. The first deleted rule is a restriction of the original rule,
so $w_0=0$.
For $1\le j<r$, Lemma~\ref{lem:ancestors} identifies $\mc D_j$
as the radius-$(r-j)$ conditional resampling of
$\widehat{\mc D}_j$.
The latter has zero conflict, occupancy $\tau_j$, and at most $m$
nonzero labels conditional on $V_j$.
Proposition~\ref{prop:direct-projection} gives
\[
 w_j\le\frac{2\tau_j\log(e m/\tau_j)}D.
\]
The final survivor $\widehat{\mc A}_r$ has zero conflict and at
most $m$ nonzero labels given the empty view $V_r$.
Its radius-zero resampling is $\mc A_r$, so the same proposition gives
\[
 \epsilon_r\le\frac{2p_r\log(e m/p_r)}D.
\]

The $r$ masses $p_r,\tau_1,\ldots,\tau_{r-1}$ sum to $p_1$.
If $p_1>0$, normalize them to a probability vector
$(\alpha_1,\ldots,\alpha_r)$.
The sum of their expressions $x\log(e m/x)$ equals
\[
 p_1\log(e m/p_1)
       +p_1\sum_{i=1}^r\alpha_i\log(1/\alpha_i)
 \le p_1\log(e m r/p_1).
\]
For $p_1=0$ every term is zero.
Summing the conflict estimates proves the first inequality.
The second follows from $p_1\le p_0$ and monotonicity of
$x\log(e m r/x)$ on $[0,1]$.
\end{proof}

The first deleted part is excluded from the entropy sum for
a specific reason. Before any resampling has taken place,
it already has zero conflict, so $w_0=0$.
It need not have a short conditional list.
Every later first-failure part has passed a previous list,
which is what supplies the $m$-coordinate condition for the
direct comparison.

Although the hatted parts are all restrictions of one
independent-set selection, they are not probabilistically
independent of one another. The proof uses their individual
zero-conflict property and the fact that their masses form
a partition. It makes no independence assumption across parts.

\Needspace{5\baselineskip}
\section{Proof of the Final Bound}
\label{sec:main-proof}

\begin{proof}[Proof of Theorem~\ref{thm:main}]
Choose an absolute $D_0$ large enough for the parameter
bounds, Proposition~\ref{prop:truncation}, and $\log D\ge1$.
All thresholds used there are uniform in the finite radius.
Start the reduction sequence with $\mc A_0=\mc A$, and put
$p=p_0$. If $p=0$, the claim is immediate.

Lemma~\ref{lem:mass-from-budget} and
Proposition~\ref{prop:joint-budget} give
\begin{equation}
 \begin{split}
 p
 &\le rb+\sqrt{36r+2}\,
             \sqrt{\epsilon_r+\sum_{j<r}w_j}\\
 &\le rb+\sqrt{\frac{76rp}{D}\log\frac{emr}{p}},
       \qquad b=\frac{2100\log D}{D},
 \end{split}
 \label{eq:final-self-bound}
\end{equation}
using $2(36r+2)\le76r$.
The right-hand occurrence of $p$ is the total mass of the
original outputs being classified. Thus this is a necessary
inequality for the unknown occupancy, not an assumption
that it is small.

If $r>D$, then $5000r\log D/D\ge1$, so the probability bound
$p\le1$ suffices. If $p<1/D$, the asserted bound also holds.
In the remaining case,
$m\le D^4$, $r\le D$, and $p\ge1/D$, whence
\[
 \log(emr/p)\le 1+6\log D\le7\log D.
\]
Set $A=532r\log D/D$. Inequality
\eqref{eq:final-self-bound} implies
\[
 p\le rb+\sqrt{Ap}
       \le rb+\frac p2+\frac A2.
\]
Subtract $p/2$ and multiply by $2$:
\[
 p\le 2rb+A
   =\frac{4732r\log D}{D}
   \le\frac{5000r\log D}{D}.
\]
Combine with $p\le1$ to finish.
\end{proof}

\section*{Acknowledgements}

We thank Shafi Goldwasser, Amir Abboud, Keren Censor-Hillel, Sepehr Assadi, Oded Goldreich, Jason Li, Christoph Grunau, Nikhil Srivastava, Antares Chen, and Peter Davies-Peck for helpful discussions.

\section{AI Usage}

Many of the key ideas in this work were developed without AI.
By December 2024, the authors had identified the hard instance discussed here
and shown that no 2-round algorithm for MIS exists on this graph.
This 2-round impossibility result used a slightly different strategy from
this work. Instead of bounding occupancy and conflict probabilities,
it bounded the support size of the algorithm's output distribution within a given clique.
The proof used Proposition \ref{prop:endpoint-tails}
and Lemma \ref{lem:branch-transfer}. The authors
wrote up this 2-round impossibility result along with an outline
for how to extend the result to $\Omega(\log \log n)$ rounds. All of this
was done without AI. Extending the approach seemed doable but tedious,
so the authors decided to take a break.

The authors decided to try again in September 2026.
They assembled one very long prompt containing both (a) the 2-round hardness result
and (b) the plan for generalizing to more rounds. After some back and forth, GPT-6
proved the $\tilde{\Omega}(D^{1/3})$-round hardness result present in this writeup.
In particular, GPT-6 came up with the key simplifying
idea: tracking occupancy and conflict probabilities instead of support,
with round elimination done via truncation and conditional resampling.
GPT-6 later produced the $\tilde{\Omega}(D^{1/2})$-round bound
and the final $\tilde{\Omega}(D)$-round bound with no additional assistance.
GPT-6 then produced a Lean 4 formalization of the proof for Theorem \ref{thm:local-main},
which is available at \url{https://github.com/aschild/mis}.

The authors also used GPT-6 Astra and GPT-6.1 Sol to draft the manuscript,
which they subsequently reviewed and revised carefully.
The authors take full responsibility for the contents of this paper.

\bibliographystyle{unsrtnat}
\bibliography{refs}

\appendix
\section{Details of the LOCAL Reduction}
\label{app:local-reduction}

\subsection{Extracting IDs}

\begin{lem}[Independent tentative identifiers]
\label{lem:local-identifiers}
Under the hypotheses of Proposition~\ref{prop:mis-obstruction},
there is a symmetric radius-one function of the edge matrices
that assigns tentative identifiers $J_{v,a}\in[n]$ to all
vertices of the clique blow-up. These identifiers are mutually
independent, and every point has probability at most $2/n$.
For any fixed clique $C_v$, the probability of a collision
among identifiers in $B_{T+2}^H(v)$ is at most $\delta_{\mathrm{ID}}$.
\end{lem}
\begin{proof}
Partition $[k]=P\sqcup Q$ with $|P|=\lfloor k/2\rfloor$.
For every oriented skeleton edge $v\to u$, flatten
$G_{vu}[P,Q]$ in coordinate order into a string $U_{v\to u}$
of length $q=\lfloor k^2/4\rfloor$.
On the same undirected edge, $U_{u\to v}$ is obtained from
$G_{vu}[Q,P]$, transposed. The two blocks are disjoint.
All $U_{v\to u}$ are therefore independent uniform strings.

Let $S_v=\bigoplus_{u\sim v}U_{v\to u}$.
Every $S_v$ is uniform, and different $S_v$ use disjoint
strings. Hence the $S_v$ are independent. XOR requires no
ordering of the branches and is invariant under
coordinate-preserving skeleton isomorphisms.
The capacity inequality $k b_{\mathrm{ID}}\le q$ permits splitting the
first $k b_{\mathrm{ID}}$ bits into $k$ independent uniform strings.
Interpret string $a$ as $X_{v,a}\in\{0,\ldots,2^{b_{\mathrm{ID}}}-1\}$ and put
$J_{v,a}=1+(X_{v,a}\bmod n)$.
Because $n\le2^{b_{\mathrm{ID}}}<2n$, every identifier has at most two
preimages and probability at most $2/2^{b_{\mathrm{ID}}}\le2/n$.
Applying separate functions to the independent strings
preserves their mutual independence.

The radius-$(T+2)$ ball has
$1+D\sum_{i=0}^{T+1}(D-1)^i$ skeleton vertices, by the girth
assumption, and thus $L$ blow-up vertices.
For two distinct vertices $x,y$,
$\Prob[J_x=J_y]\le2/n$ by independence and the point-probability
bound. Union bound over $\binom L2$ pairs gives $\delta_{\mathrm{ID}}$.
No independence between identifiers and graph edges is claimed.
\end{proof}

\subsection{Producing Rules}

\begin{lem}[The repaired simulation]
\label{lem:local-simulation}
For every fixed identifier-indexed tape table $\mc W$,
the candidate and repair construction in
Proposition~\ref{prop:mis-obstruction} defines a symmetric
radius-$(T+3)$ edge-input rule on the infinite $D$-regular
tree. The rule is independent on every matrix realization.
Its occupancy at every fixed clique of the finite skeleton
$H$ is at most
$\min\{1,5000(T+3)\log D/D\}$.
\end{lem}
\begin{proof}
At clique $C_w$, first inspect the tentative identifiers in
$B_{T+1}^H(w)$. If any collide, output no candidates.
Otherwise simulate the $T$ rounds of $\mc A$ at every vertex
of $C_w$, using the tape indexed by its identifier and
port numbers given by the increasing order of neighbor
identifiers. All candidate bits are defined before repair.

A graph path of length $T$ traverses at most $T$ skeleton
edges. Matrices internal to $B_{T+1}^H(w)$ give complete
adjacency lists and degrees for vertices over $B_T^H(w)$.
Together with the identifiers there, they also determine
all incident port numbers needed by the simulation, including
boundary vertices' ports toward their predecessors.
This supplies the full initial data at graph distance at
most $T$, even when neighbor identifiers are given initially.
Each tentative identifier requires one additional skeleton
layer, so candidates depend only on a radius-$(T+2)$ matrix
view. Retaining a candidate only if it has no candidate
neighbor costs one more layer.

Adjacent candidates cannot both survive, and a clique contains
at most one survivor. Thus return its coordinate, or $0$
when there is no survivor. This argument does not use
correctness of $\mc A$ on the tentative input.
All operations are invariant under permutations of skeleton
branches. The tape table is a fixed part of the rule, with
no skeleton vertex distinguished by it.

The algorithm's local transition functions may be extended
arbitrarily to malformed inputs if necessary. Whenever a
candidate gate passes, the local data relevant to its output
can instead be completed to a legal $n$-vertex input: this
tree-shaped neighborhood appears in $H$, and $kN\le n$.
Locality makes the simulated output independent of that
completion. In particular, the proof assumes no correctness
on duplicate identifiers or infinite graphs.

Apply Theorem~\ref{thm:main} on the infinite tree.
Since the girth of $H$ is at least $2(T+3)+2$, every
radius-$(T+3)$ skeleton ball is an induced tree, with the
same matrix distribution as its infinite-tree counterpart.
The same occupancy bound follows at every fixed clique of $H$.
\end{proof}

\subsection{Extracting a Contradictory ID Assignment}

\begin{lem}[Local completion and the tape table]
\label{lem:local-completion}
Fix $C_v$. For every matrix realization, one can choose a legal
identifier assignment $I^v$ on the padded $n$-vertex graph,
independently of the tape table, such that on the event of
distinct tentative identifiers in $B_{T+2}^H(v)$,
the repaired edge-input output at $C_v$ agrees with the
repaired output of $\mc A$ on this legal input.
Conditional on each matrix realization, the assigned private
tapes have the distribution of the original randomized algorithm.
\end{lem}
\begin{proof}
On the stated event, preserve all identifiers in
$B_{T+2}^H(v)$ and assign the unused identifiers bijectively
to the remaining vertices. Outside the event, use any fixed
bijection. Both choices may use a fixed external enumeration
of vertices; this enumeration is never supplied to the
edge-input rule. Order all ports by the completed identifiers.
The assignment depends only on the matrices.

For $w=v$ or $w\sim v$,
$B_{T+1}^H(w)\subseteq B_{T+2}^H(v)$.
Consequently all candidate gates relevant to repair at $C_v$
pass, and the necessary initial data and tapes agree with
those on the completed input. Locality gives the asserted
agreement of outputs.

The table entries $W_i$, $i\in[n]$, are independent with the
original tape distribution and independent of the matrices.
The completed identifier assignment is bijective, so,
conditional on every matrix realization,
$(W_{I^v(x)})_x$ is a permutation of this independent table.
It therefore has exactly the private-tape distribution of
$\mc A$. No bound on tape length is needed.
\end{proof}

\subsection{Maximality and Finite Skeletons}

\begin{lem}[External domination]
\label{lem:local-domination}
For a fixed clique $C_v$, the probability that some choice
of at most one vertex from each neighboring clique dominates
all of $C_v$ is at most $\gamma_{\mathrm{EXT}}$.
Outside this graph event, every MIS occupies $C_v$.
\end{lem}
\begin{proof}
There are $(k+1)^D$ choices, including the empty choice
in each neighboring clique. For a fixed choice of $j\le D$
vertices, every vertex of $C_v$ is undominated with probability
$2^{-j}$, independently across the $k$ rows. Thus the chance
that this choice dominates every row is at most
$(1-2^{-j})^k\le\exp(-k2^{-D})$.
Union bound gives $\gamma_{\mathrm{EXT}}$.
An independent set takes at most one vertex from each
neighboring clique. If it leaves $C_v$ empty outside this
event, a vertex of $C_v$ has no selected neighbor, contradicting
maximality. The graph event quantifies over every external
choice, so the selected vertices may depend on the matrices.
\end{proof}

\begin{lem}[Regular skeletons of prescribed size; folklore]
\label{lem:local-girth}
For integers $D,g\ge3$, put
$B=1+D\sum_{i=0}^{g-3}(D-1)^i$.
Every even integer $N>4B$ admits a simple $D$-regular
graph with girth at least $g$. Moreover, $B\le2D^{g-2}$.
\end{lem}
\begin{proof}
Among simple graphs on $N$ fixed vertices with maximum degree
at most $D$ and girth at least $g$, choose one with the largest
number of edges. Suppose it is not regular. Total deficiency
$DN-\sum_x\deg(x)$ is a positive even integer.
Choose deficient vertices $u,v$; if there is only one, take
$u=v$, whose deficiency is then at least two.
Every deficient vertex other than $u$ is at distance at most
$g-2$ from $u$, since otherwise its edge to $u$ could be added.

Let $U=B_{g-2}(u)\cup B_{g-2}(v)$, so $|U|\le2B$.
Every vertex outside $U$ has degree $D$.
If there were no edge $xy$ with both endpoints outside $U$,
then all $D|V\setminus U|$ incidences outside would lead into
$U$, giving $N\le2|U|\le4B$, a contradiction.
Delete such an edge $xy$ and add $ux,vy$; when $u=v$, add
$ux,uy$. This preserves the degree bound and adds one edge.
These new edges are not loops or parallel edges.

A cycle containing only one new edge would give an old path
from $u$ or $v$ to $x$ or $y$ of length below $g-1$, contrary
to their distances. A cycle using both new edges either
contains an old $x$--$y$ path in the graph with $xy$ deleted,
or contains two old paths from $\{u,v\}$ to $\{x,y\}$.
In the first case that path has length at least $g-1$ by the
old girth; in the second case both paths do by distance.
The same first case covers $u=v$.
Hence no cycle of length below $g$ is introduced, contradicting
maximality. The graph was regular.
Finally, the maximum-degree ball bound gives
$B\le\sum_{i=0}^{g-2}D^i\le2D^{g-2}$.
\end{proof}

\subsection{Reduction Parameters}

\begin{lem}[Parameter checks for the LOCAL lower bound]
\label{lem:local-parameters}
Fix $0<\sigma\le1$. Let $n$ tend to infinity, and let $T$ be
any nonnegative integer satisfying
$T<\sigma\log n/(100000\log\log n)$.
For the parameters $D,k,N,g$ in the proof of
Theorem~\ref{thm:local-main}, all the graph and parameter hypotheses of
Proposition~\ref{prop:mis-obstruction} hold, for a suitable
skeleton $H$, for sufficiently large $n$. Its error terms satisfy
$\delta_{\mathrm{ID}}=o(1)$ and $n\gamma_{\mathrm{EXT}}=o(1)$,
uniformly over these choices of $T$.
\end{lem}
\begin{proof}
We have $D=(\log n)/4+O(1)$ and $k=n^{1/4+o(1)}$.
In particular $D\ge D_0$, $D\le\log k\le D^2$, and
$kN\le n$ for large $n$. Since
$b_{\mathrm{ID}}=O(\log n)$ while $k$ is a positive power
of $n$, also $k b_{\mathrm{ID}}\le\lfloor k^2/4\rfloor$.

The integer $N$ is even and $N=n^{3/4+o(1)}$.
For the ball bound $B$ in Lemma~\ref{lem:local-girth},
\[
 \log(4B)\le\log 8+(2T+6)\log D
 \le(2\sigma/100000+o(1))\log n.
\]
Thus $N>4B$, and that lemma supplies the required
$D$-regular skeleton of girth at least $g=2T+8$.
The local ball contains at most $2D^{T+2}$ skeleton vertices,
so $L\le2kD^{T+2}$ and
\[
 \delta_{\mathrm{ID}}\le L^2/n
 \le n^{-1/2+2\sigma/100000+o(1)}=o(1).
\]
Finally, $k+1\le e^{D+1}$ and $k\ge e^D$ for large $D$.
Hence
\[
 \log(n\gamma_{\mathrm{EXT}})
 \le\log n+D(D+1)-e^{(1-\log2)D}\longrightarrow-\infty.
\]
All remainders and thresholds above depend only on $\sigma$
and the absolute constants. The stronger estimate
$n\gamma_{\mathrm{EXT}}=o(1)$ also justifies the introductory
claim that, with high probability, every MIS intersects every
clique simultaneously: union bound over $N\le n$ cliques.
\end{proof}
\end{document}